\documentclass[journal]{IEEEtran}
\IEEEoverridecommandlockouts
\usepackage{cite}
\usepackage{amsmath,amssymb,amsfonts}
\usepackage{graphicx}
\usepackage{xspace}
\usepackage{textcomp}
\usepackage{xcolor}
\usepackage{url}
\usepackage{algpseudocode}
\usepackage{algorithm}
\usepackage{footnote}
\usepackage{subcaption}
\usepackage{float}
\usepackage{mathtools}
\usepackage{amsthm}
\usepackage{paralist}
\usepackage{etoolbox}
\theoremstyle{remark}
\newtheorem{remark}{Remark}[section]

\algnewcommand\algorithmicforeach{\textbf{for each}}
\algdef{S}[FOR]{ForEach}[1]{\algorithmicforeach\ #1\ \algorithmicdo}

\makeatletter
\patchcmd{\@makecaption}{\scshape}{}{}{}
\makeatother

\newtheorem{theorem}{Theorem}[section]
\newtheorem{lemma}[theorem]{Lemma}
\newtheorem{definition}{Definition}[section]

\def\BibTeX{{\rm B\kern-.05em{\sc i\kern-.025em b}\kern-.08em
    T\kern-.1667em\lower.7ex\hbox{E}\kern-.125emX}}

\begin{document}

\title{Masking Radar Cognition under Adversarial Surveillance:
A Distributional Privacy Framework}

\author{\IEEEauthorblockN{Sreedevi K\IEEEauthorrefmark{1}, Nandhini K\IEEEauthorrefmark{2}, Anup Aprem\IEEEauthorrefmark{1}, Deepthi P P\IEEEauthorrefmark{1}}\\
\IEEEauthorblockA{\IEEEauthorrefmark{1}\textit{Department of Electronics and Communication Engineering}\\
\textit{National Institute of Technology Calicut, India}\\
\url{{sreedevi_p230653ec, anup.aprem, deepthi}@nitc.ac.in}}\\
\IEEEauthorblockA{\IEEEauthorrefmark{2}\textit{Department of Electrical Engineering}\\
\textit{Indian Institute of Technology Madras, Chennai, India}\\
\url{ee25d005@smail.iitm.ac.in}}}

\maketitle

\begin{abstract}
In this article, we propose an online electronic
counter-countermeasure (ECCM) framework designed to conceal
the strategic decision-making processes of a cognitive radar
(CR) operating under adversarial surveillance. We model the CR
under two distinct decision paradigms: a static constrained
utility-maximizing behavior and a dynamic expected
utility-maximizing behavior. The radar's utility function is
modeled via a von Mises--Fisher (vMF) distribution, with the
distributional parameter constituting the private information
to be protected from adversarial inference. 
We adopt a
distribution privacy framework to conceal this private
information and provide formal distribution privacy guarantees
for cognition masking. In this work, we develop cognition-hiding
algorithms for both static constrained utility maximization
(WDPCH-SU), and dynamic expected utility maximization
(WDPCH-DU). Through rigorous mathematical analysis, we show
that both WDPCH-SU and WDPCH-DU satisfy $\epsilon$-distribution
privacy ($\epsilon$-DistP) against inference-based adversarial
attacks and present the privacy--performance trade-off bounds,
quantifying utility loss (in static setting) and expected
utility deviation (in dynamic setting) as functions of
$\epsilon$. Numerical results show that WDPCH-SU gives about
15\% improvement in utility loss at maximum privacy compared to
the existing methodology while WDPCH-DU achieves a greater
reduction in adversarial Fisher information without requiring
explicit Fisher information constraints, at a moderate,
analytically bounded utility deviation. These results are
highly promising in many 6G communication scenarios such as
network slicing for automated driving and swarm UAV
coordination, where it is essential to keep the resource
allocation policy robust against privacy attacks.
\end{abstract}


\section{Introduction}
\IEEEPARstart{T}{raditionally}, radar systems operated using
predefined search patterns and periodic track updates, with
limited ability to adapt to dynamic environments. The emergence
of cognitive radar (CR)~\cite{Haykin2006} has transformed this
paradigm by incorporating learning and adaptation into the
radar decision-making process, allowing the radar to refine its actions based
on past observations and perception of the environment. This
adaptive capability improves system performance and extends its
applicability to complex and uncertain electromagnetic
environments~\cite{Griffiths2020}, optimal waveform
selection~\cite{Mitchell2016}, efficient resource
management~\cite{Bell2015}, and robust multi-target
tracking~\cite{UbedaMedina2016}.

The increasing deployment of CR systems has also accelerated
advancements in electronic countermeasures
(ECM)~\cite{Zhang2022} and electronic
counter-countermeasures (ECCM)~\cite{Gupta2022}. In this paper, we consider an adversarial environment, wherein an intelligent adversary
observes the actions of a CR system over time and exploit
this observable behavior to estimate the underlying model parameters, consequently deducing the decision-making strategy, a process referred to in this
paper as \emph{inverse learning}. Only a few works in the literature explicitly
address ECCM against inverse learning. For
example,~\cite{JainRadarConf2024,JainTAES2025} developed
algorithms that reduce the Fisher information available to an
adversarial estimator. However, these methods lack formal
statistical guarantees to protect the internal decision-making
processes of CR systems.

In this work, we address this limitation by formulating the
ECCM problem in a distributional privacy setting. The radar is modeled under two complementary paradigms: a static constrained
utility-maximizing behavior and a dynamic expected
utility-maximizing behavior modelled using a long-run average
reward Markov Decision Process (MDP). \emph{To the best of our
knowledge, ours is the first work to introduce formal
distribution privacy guarantees for cognition masking in
cognitive radar ECCM.} The main contributions of this article
are as follows:
\begin{enumerate}

\item We develop WDPCH-SU, a cognition-hiding algorithm for
a static constrained utility-maximization setting---applicable
to single-stage resource allocation tasks such as beam
allocation, waveform selection, and power allocation.
\item We develop WDPCH-DU, a cognition-hiding algorithm for
dynamic expected utility-maximization over a long-run average reward
Markov decision process---applicable to tasks such as
operating-mode selection and SINR-based action
planning.
\item We prove that both WDPCH-SU and WDPCH-DU satisfy
$\epsilon$-distribution privacy ($\epsilon$-DistP), providing
formal guarantees against inference-based adversarial attacks
on the radar's decision-making process.
\item We derive analytical bounds characterizing the
privacy--performance trade-off, quantifying utility loss in
the static setting and expected utility deviation in the
dynamic framework, as a function of the
privacy budget $\epsilon$.

\item Numerical results demonstrate that WDPCH-SU achieves
lower utility loss and stronger privacy than existing
methodology~\cite{Pattanayak2023} for the same privacy
budget, while WDPCH-DU achieves a larger reduction in adversarial
Fisher information than existing approaches,
without requiring explicit Fisher information
constraints~\cite{JainTAES2025} or transition probability
perturbation~\cite{GohariCDC2020}---both with minimal
modifications to baseline radar controllers.

\end{enumerate}

\textit{Organization of the article:} The remainder of this
article is organized as follows. Section~II reviews the related
literature. Section~III introduces the concept of
distributional privacy and the theoretical foundations used in
this work. Sections~IV and~V develop
the proposed ECCM algorithms, namely WDPCH-SU and WDPCH-DU,
respectively. Section~VI presents numerical results.
Finally, Section~VII offers concluding remarks.

\section{Related Work}\label{sec:related}

The evolution of electronic counter-countermeasures (ECCM) in
radar systems reflects a shift from mitigating signal-level
interference to safeguarding decision-making processes against
adversarial inference. With the emergence of cognitive radar,
adversaries increasingly seek to infer internal objectives,
state estimates, and decision strategies from observable
behavior~\cite{Haykin2003,Haykin2006,Haykin2012,Bell2015},
necessitating ECCM paradigms beyond traditional signal
processing techniques.

\subsection{Classical ECM and ECCM}
Classical electronic countermeasures (ECM) degrade radar
performance through noise and deceptive jamming, and correspondingly, classical ECCM techniques enhance robustness at the signal
level. Representative ECCM approaches include spread-spectrum
methods such as Direct Sequence Spread Spectrum (DSSS) and
Frequency Hopping Spread Spectrum (FHSS) for mitigating
narrowband interference~\cite{Levanon2004,Skolnik2008},
space--time adaptive processing (STAP) for suppressing
non-stationary interference~\cite{Goldstein1997,Wicks2006},
and beamforming for spatial interference
suppression~\cite{Richards2010}. However, these methods do not
address vulnerabilities arising from inference-based attacks on
radar decision processes.

\subsection{Cognitive Radar: ECM and ECCM}
The introduction of cognition into radar systems has led to
more sophisticated adversarial strategies, broadly categorized
as cognitive jamming and inverse learning. Cognitive jamming
exploits radar behavior using learning and optimization
frameworks, including game-theoretic formulations for
radar--jammer interaction~\cite{Song2012,Gao2015},
principal--agent models capturing information
asymmetry~\cite{Gupta2022}, and reinforcement learning-based
adaptive jamming strategies~\cite{LiJiuLiu2019,Zhang2023}.
Inverse learning attacks aim to infer an agent's internal model
from observed actions~\cite{Ng2000,Abbeel2004}, enabling
reconstruction of internal state
estimates~\cite{Krishnamurthy2020SPIN,Krishnamurthy2020} or
decision-making policies and
objectives~\cite{Krishna2020Rd,AnoopAprem2025}. Unlike
jamming, inverse learning constitutes a persistent threat by
enabling long-term exploitation of radar behavior.

Complementary to these adversarial ECM techniques, ECCM
strategies include adaptive sensor
selection~\cite{Dai2020}, reinforcement learning-based
anti-jamming~\cite{Jiang2023,Thornton2021}, and adversarial
bandit and game-theoretic
formulations~\cite{Howard2022,Kang2023}. While these
approaches improve robustness against intelligent interference,
they primarily focus on performance optimization and do not
explicitly address adversarial inference of radar
decision-making policies.

\subsection{Metacognitive Radar and Masking Strategies}

To mitigate inference-based threats, metacognitive radar
frameworks have been proposed~\cite{Pattanayak2023}, wherein
the radar perturbs its action strategy to obscure underlying
objectives and hinder adversarial learning. More broadly,
masking of action strategies has been explored across domains
via deception and randomized control policies, including
KL-divergence-based stealth attacks~\cite{Zhang2020},
control-theoretic deception~\cite{Karabag2022}, inverse
reinforcement learning and revealed preference-based
masking~\cite{Pattanayak2023}, entropy-maximizing
policies~\cite{Savas2020}, and Fisher information
minimization~\cite{Karabag2019}. In the context of radar
systems, action-plan masking~\cite{JainRadarConf2024} and
Fisher information-based approaches~\cite{JainTAES2025} have
been investigated, wherein the radar purposefully reduces the
determinant of the adversary's Fisher Information Matrix
(FIM)~\cite{Kay1993,Li2004,Shirazi2019,Soen2021} to increase
adversarial estimation uncertainty. While effective
empirically, none of these approaches---whether
heuristic masking, revealed preference-based perturbation,
or FIM minimization---provide formal statistical guarantees
on what an adversary can infer about the radar's
decision-making parameters. They are further limited by
reliance on specific adversarial assumptions and, in the
case of FIM-based formulations, non-convex optimization
that precludes tractable worst-case analysis.

\subsection{Privacy Frameworks}
Differential privacy (DP)~\cite{Dwork2006,DworkRoth2014}
provides a principled framework for limiting information
leakage in statistical systems, offering worst-case guarantees
on what can be inferred about any individual data point.
Distribution privacy~\cite{Kawamoto2019} extends this notion
to the distributional level, protecting the parameters
governing the underlying distribution rather than individual
records---useful in settings where the data-generating
distribution itself is the sensitive quantity. However,
conventional distribution privacy mechanisms often incur
significant utility loss. The Wasserstein
mechanism~\cite{SongPrivacy2017,Chen2023} achieves an improved
utility--privacy trade-off while preserving distribution-level
guarantees. These frameworks have also been extended to
sequential decision-making, with privacy-preserving policy
synthesis in MDPs explored in~\cite{GohariCDC2020}.

Existing literature addresses signal-level ECCM, radar--jammer
interactions, inverse learning attacks, and heuristic masking
strategies, yet a framework for cognition hiding with formal
statistical guarantees remains lacking.

\section{Background: Distribution Privacy and Wasserstein Mechanisms}

Distribution privacy provides a principled framework for
obfuscating the radar's cognitive strategy.

\begin{definition}[$\varepsilon$-Distribution Privacy~\cite{Chen2023}]
\label{def:4.1}
A radar's decision-making satisfies $\varepsilon$-distribution
privacy if, for any dataset $\mathcal{D}$ observed by the
adversary and any two neighboring parameters $\theta_i$ and
$\theta_j$, the likelihood ratio of $\mathcal{D}$ under the
corresponding distributions $P_{\theta_i}$ and $P_{\theta_j}$
satisfies
\begin{equation}
\frac{P(\mathcal{D} \mid P_{\theta_i})}
{P(\mathcal{D} \mid P_{\theta_j})}
\leq e^{\varepsilon}.
\label{eqn:dist:privacy:def}
\end{equation}
We refer to this guarantee as \emph{$\varepsilon$-distribution
privacy} ($\varepsilon$-DistP).
\end{definition}

$\varepsilon$-DistP guarantees that, with probability at least
$1 - e^{-\varepsilon}$, an adversary observing $\mathcal{D}$
cannot reliably distinguish whether it was generated under
$P_{\theta_i}$ or $P_{\theta_j}$, thereby fundamentally
limiting its ability to infer the latent cognitive parameter
$\theta$ governing the radar's decision-making process.

\begin{definition}[$\infty$-Wasserstein Distance]
\label{def:4.2}
For two probability distributions $\mu$ and $\nu$ in
$\mathbb{R}^n$, let $\Gamma(\mu, \nu)$ represent the set of
all joint distributions with marginals $\mu$ and $\nu$. The
$\infty$-Wasserstein distance $W_{\infty}(\mu, \nu)$ between
$\mu$ and $\nu$ is given by:
\begin{equation}
W_{\infty}(\mu, \nu) = \adjustlimits\inf_{\gamma \in
\Gamma(\mu, \nu)} \sup_{(x,y) \in \text{supp}(\gamma)}
\|\boldsymbol{x} - \boldsymbol{y}\|_1.
\end{equation}
\end{definition}

The $\infty$-Wasserstein distance focuses on the worst-case
cost required to map one distribution to another, making it
useful in privacy and robustness scenarios where the most
extreme deviations define the system sensitivity. The maximum
$\infty$-Wasserstein distance over a collection of
distributions $\Psi$ is
\begin{equation}
\Delta_W(\Psi) = \sup_{(P_{\theta_i}, P_{\theta_j}) \in \Psi}
W_{\infty}(P_{\theta_i}, P_{\theta_j}).
\label{eqn:delta:W}
\end{equation}

\section{Masking the Static Utility of a Cognitive Radar}
\label{sec:cogradar}

\begin{figure*}[t]
    \centering
    \includegraphics[width=\linewidth]{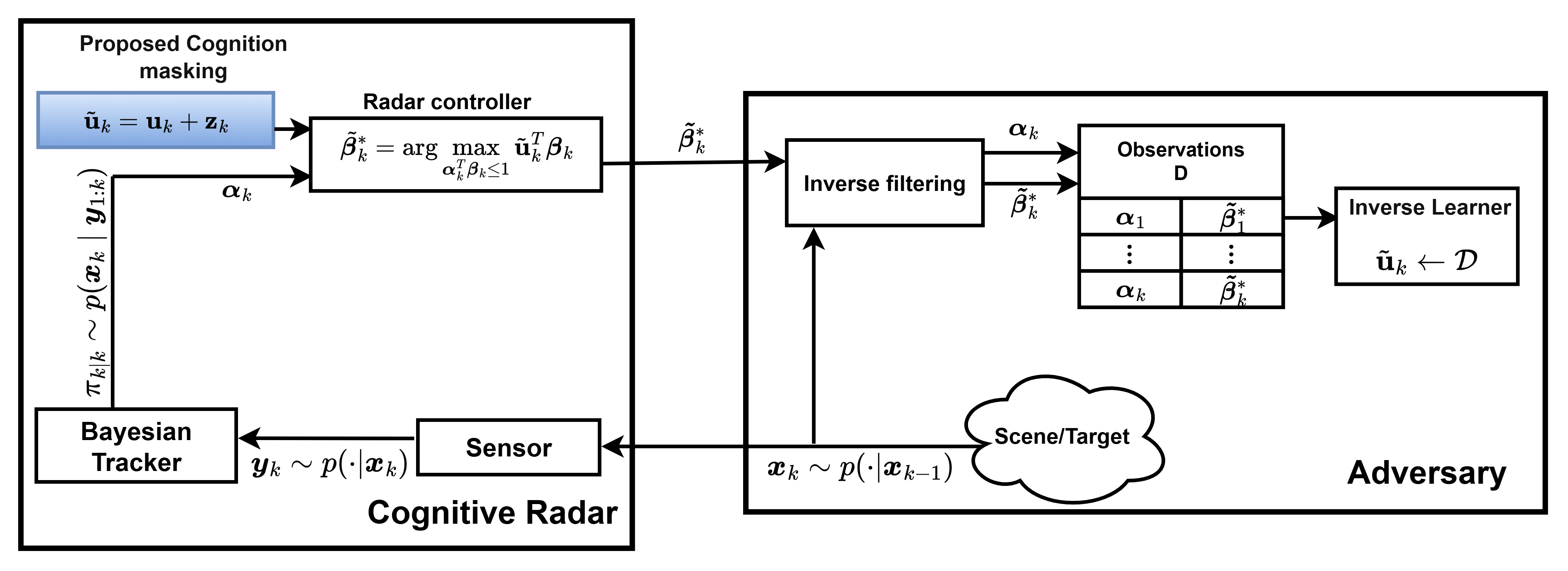}
  \caption{Cognitive radar--adversary interaction under partial
observability. The Bayesian tracker updates the belief
$\boldsymbol{\pi}_{k|k} = p(\mathbf{x}_k \mid \mathbf{y}_{1:k})$ from sensor observations and produces
the scene perception $\boldsymbol{\alpha}_k$. The radar
controller applies the proposed cognition masking mechanism,
perturbing the utility as $\tilde{\mathbf{u}}_k = \mathbf{u}_k +
\mathbf{z}_k$, and computes the privacy-preserving response
$\tilde{\boldsymbol{\beta}}_k^* = \arg\max_{\boldsymbol{\beta}_k}
\tilde{\mathbf{u}}_k^T\boldsymbol{\beta}_k$ subject to
$\boldsymbol{\alpha}_k^T\boldsymbol{\beta}_k \leq 1$. The
adversary probes the radar through targeted maneuvers of the
scene/target, generating $\mathbf{x}_k \sim
p(\cdot \mid \mathbf{x}_{k-1})$, and observes the
radar's response $\tilde{\boldsymbol{\beta}}_k^*$. Using
$\mathbf{x}_k$ and $\tilde{\boldsymbol{\beta}}_k^*$, the
adversary applies inverse filtering to recover the scene
perception $\boldsymbol{\alpha}_k$, accumulating the dataset
$\mathcal{D} = \{(\boldsymbol{\alpha}_k,
\tilde{\boldsymbol{\beta}}_k^*)\}_{k=1}^K$. The inverse learner
then estimates the radar's utility $\hat{\mathbf{u}}_k$ from
$\mathcal{D}$, which the WDPCH-SU mechanism conceals
from adversarial inference.}
    \label{fig:radar_policy}
\end{figure*}

\subsection{CR Framework for Static Utility Maximization}
\label{sec:static_framework}

In this section, we consider a mathematical framework to model the interaction
between a cognitive radar (CR) and an intelligent adversary
in an electronic warfare environment. The CR aims to track a
target, such as a drone, UAV, or aircraft, while the adversary
monitors the radar's actions to infer the underlying
decision-making parameters governing radar behavior. To
characterize this interaction, we adopt the widely used Fully
Adaptive Radar (FAR) framework~\cite{Bell2015,Mitchell2016},
shown in Fig.~\ref{fig:radar_policy} and adapted
from~\cite{Krishnamurthy2020,Krishna2020Rd,AnoopAprem2025}.
The framework comprises a scene describing the target
dynamics, a sensor module, a Bayesian tracker, and a
perception-driven radar controller that selects optimal actions by maximizing an associated utility function. Fig.~\ref{fig:radar_policy}
also depicts an adversary equipped with an inverse learner,
which observes the radar's actions to infer the
underlying decision-making parameters. As a canonical example,
we consider the scenario in Fig.~\ref{Fig:radar-target}, where
a cognitive radar tracks multiple targets and allocates its
beam resources among them based on their perceived tracking
accuracy, while an adversary manipulates target maneuvers and
observes the resulting allocations to infer the radar's
cognitive strategy. In the following, we define the various components of the
FAR framework.

The scene consists of a moving target whose kinematic state
evolves in discrete time, $k = 1, 2, \ldots$. Let
$\mathbf{x}_k \in \mathcal{X}$ denote the target kinematic
state at time $k$. Assuming Markovian dynamics, the state
evolves according to
\begin{equation}
\text{{Motion dynamics:}}\quad
\mathbf{x}_k \sim p(\cdot \mid \mathbf{x}_{k-1}).
\label{eq:motion_model}
\end{equation}
The sensor comprises a transceiver that illuminates the
environment and records the reflected signal. The resulting
observation $\mathbf{y}_k \in \mathcal{Y}$ is a noisy
measurement of the target state $\mathbf{x}_k$, governed by
the conditional density
\begin{equation}
\text{{Observation:}}\quad
\mathbf{y}_k \sim p(\cdot \mid \mathbf{x}_k).
\label{eq:obs_model}
\end{equation}

The cognitive radar employs a Bayesian tracker---such as a
Kalman filter for linear Gaussian models~\cite{Haykin2006,Haykin2012}
or a particle filter for nonlinear and non-Gaussian
models~\cite{Krishnamurthy2019,Singh2022}---to compute a
posterior belief over the target state. Let
$\boldsymbol{\pi}_{k|k}$ denote the posterior belief at
time $k$:
\begin{equation}
\text{{Perception:}}\quad
\boldsymbol{\pi}_{k|k} = p(\mathbf{x}_k \mid \mathbf{y}_{1:k}),
\label{eq:perception}
\end{equation}
where $\mathbf{y}_{1:k} = \{\mathbf{y}_1, \ldots,
\mathbf{y}_k\}$ denotes the accumulated measurement sequence
up to time $k$.
From this belief, the tracker
generates the \emph{scene perception vector} $\boldsymbol{\alpha}_k: \boldsymbol{\pi}_{k|k} \to \mathbb{R}^n$. For the example in Fig.~\ref{Fig:radar-target},  
$\boldsymbol{\alpha}_k= [\alpha_k(1), \ldots,
\alpha_k(n)]^T$, whose $i$-th entry
\begin{equation}
\alpha_k(i) = \mathrm{Tr}\!\left(\boldsymbol{\Sigma}_{k|k-1}^{-1}(i)\right)
\label{eq:alpha_def}
\end{equation}
is the predicted tracking accuracy for target
$i$~\cite{AnoopAprem2025,Krishna2020Rd}, where $\mathrm{Tr}(\cdot)$
denotes the matrix trace. For instance, under a Kalman filter,
the belief is parameterized as $\boldsymbol{\pi}_{k|k} \sim
\mathcal{N}\left(\hat{\mathbf{x}}_{k|k},
\boldsymbol{\Sigma}_{k|k}\right)$, and
$\boldsymbol{\Sigma}_{k|k-1}(i)$ is the $i$-th diagonal block of
the predicted error covariance
$\mathrm{Cov}(\boldsymbol{\pi}_{k|k-1})$, with
$\boldsymbol{\pi}_{k|k-1}$ obtained by propagating
$\boldsymbol{\pi}_{k-1|k-1}$ through the motion dynamics
in~\eqref{eq:motion_model}.

Conditioned on the scene perception $\boldsymbol{\alpha}_k$,
the radar controller selects an action, for example beam allocation in Fig.~\ref{Fig:radar-target},
$\boldsymbol{\beta}_k^* \in \mathbb{R}_+^n$---the
\emph{response signal}---by solving the constrained
utility-maximization problem
\begin{align}
\text{Controller policy:}\quad
\boldsymbol{\beta}^* &= \arg\max_{\boldsymbol{\beta}}\;
U_k(\boldsymbol{\beta})
\label{eq:policy}\\
&\text{s.t.}\quad
\boldsymbol{\alpha}_k^T\boldsymbol{\beta} \leq 1,\nonumber\\
&\phantom{\text{s.t.}\quad}
0 \leq \beta_i \leq 1,\quad \forall\, i,\nonumber
\end{align}
where $U_k(\boldsymbol{\beta})$ is the utility function governing the radar's relative preference across the $n$ targets.

\begin{figure}[t]
    \centering
    \includegraphics[width=0.95\linewidth]{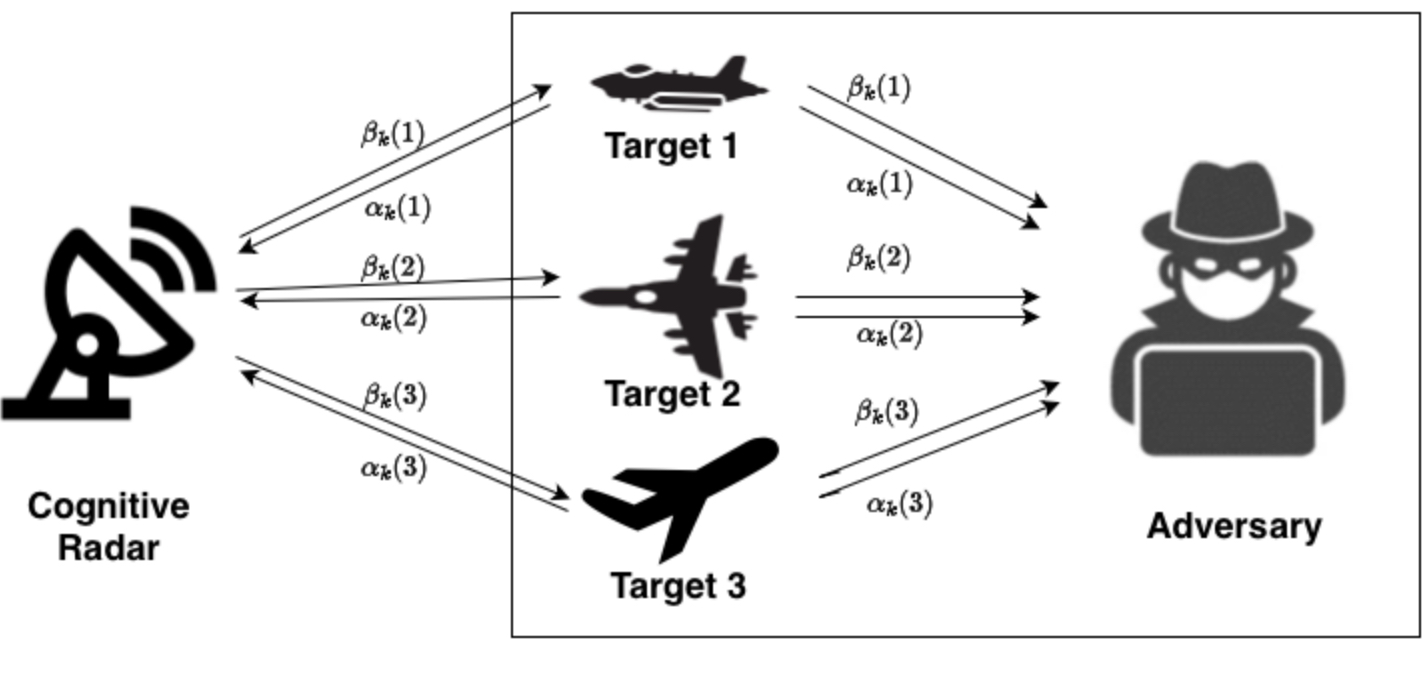}
    \caption{Electronic warfare scenario where a metacognitive
    radar tracks multiple adversarial targets. An adversary
    manipulates target maneuvers to infer the radar's cognitive
    strategy. As an ECCM measure, the radar hides its cognition
    during normal operation.}
    \label{Fig:radar-target}
\end{figure}

\subsection{Adversarial Inverse Learning}
Fig.~\ref{fig:radar_policy} illustrates the proposed inverse
learning framework for adversarial cognition inference. At
time $k$, the adversary probes the radar through deliberate
target maneuvers---such as changes in acceleration---which
generate the scene/target state $\mathbf{x}_k \sim
p(\mathbf{x}_k \mid \mathbf{x}_{k-1})$, and observes the
resulting radar response $\boldsymbol{\beta}_k^*$. Using
$\mathbf{x}_k$ and $\boldsymbol{\beta}_k^*$, and assuming the
motion dynamics in~\eqref{eq:motion_model} are known, the
adversary's inverse learner performs inverse filtering~\cite{Krishnamurthy2020,Krishnamurthy2020SPIN} to
compute the corresponding scene perception vector
$\boldsymbol{\alpha}_k$. Over time, this yields the dataset
$\mathcal{D} = \{(\boldsymbol{\alpha}_k,
\boldsymbol{\beta}_k^*)\}_{k=1}^K$, consisting of the sequence
of recovered perception vectors $\boldsymbol{\alpha}_k$
together with the corresponding radar responses
$\boldsymbol{\beta}_k^*$, accumulated over $K$ time steps.
Based on $\mathcal{D}$, the adversary aims to infer the
underlying utility function $U_k(\cdot)$. If exposed, this
inferred utility can be exploited by the adversary to design
countermeasures that degrade the radar's tracking performance.

\subsection{Stochastic Linear Utility Model}
\label{sec:stochastic_util}
In this paper, we restrict our attention to a linear utility
function~\cite{AnoopAprem2025,Krishna2020Rd}, i.e.,
$U_k(\boldsymbol{\beta}) = \mathbf{u}_k^T\boldsymbol{\beta}$. In
addition, we assume a stochastic prior for $\mathbf{u}_k$. This
follows the cognitive radar utility model
of~\cite{AnoopAprem2025,Krishna2020Rd}, on which our framework
is built, and reflects the fact that the radar's relative
preference across targets may vary mildly across decision
instants due to sensor noise, clutter, and short-term
fluctuations in perceived target priority, even under a stable
long-run objective.
Substituting the linear utility into~\eqref{eq:policy}, the
controller's constrained optimization problem reduces to
\begin{align}
\boldsymbol{\beta}_k^* &= \arg\max_{\boldsymbol{\beta}_k}\;
\mathbf{u}_k^T\boldsymbol{\beta}_k,
\label{eq:LP-unperturbed}\\
&\text{s.t.}\quad
\boldsymbol{\alpha}_k^T\boldsymbol{\beta}_k \leq 1,\quad
0 \leq \beta_{k,i} \leq 1,\quad \forall\, i.\nonumber
\end{align}
Here, \( \boldsymbol{u}_k = [ {u}_k(1),  {u}_k(2), \ldots,
 {u}_k(n)]^T \in \mathbb{R}^n\) is the utility vector, whose
elements \( {u}_k(i) \) quantify the relative preferences of
the radar. For example, in the beam allocation problem, as in
Fig.~\ref{Fig:radar-target}, while tracking~\( n \) different
targets, \( {u}_k(i) \) is the incentive gained by the CR for
allocating a unit time duration beam to target~\( i \). As defined in~\eqref{eq:LP-unperturbed},
\(\boldsymbol{\alpha}_k = [{\alpha}_k(1), {\alpha}_k(2), \ldots,
{\alpha}_k(n)] \in \mathbb{R}_+^n\), where~$\boldsymbol{\alpha}_k(i)$
quantifies the tracking accuracy of target $i$ at time $k$, and
\(\boldsymbol{\beta}_k = [{\beta}_k(1), {\beta}_k(2), \ldots,
{\beta}_k(n)]^T \in \mathbb{R}_+^n\) denotes the radar action,
i.e., the fraction of time allocated to each target.

Following~\cite{AnoopAprem2025}, we assume that \(\boldsymbol{u}_k
\) follows some distribution~$\mathcal{P}_u$. Due to the
ordinality\footnote{For example, if the utility vector
in~\eqref{eq:LP-unperturbed} is scaled by a positive constant,
the result of the optimization problem is unchanged} of the
utility vector in~\eqref{eq:LP-unperturbed}, the distribution
$\mathcal{P}_u$ can be restricted to the unit hyper-sphere
$\mathcal{S}^{n-1} = \{\boldsymbol{u}_k : ||
{\boldsymbol{u}_k }||_2 = 1\}$. One suitable candidate for
$\mathcal{P}_u$ from Gaussian family is the von Mises Fisher
(vMF) distribution~\cite{Mardia2009}. The
probability density function of the vMF distribution for a
unit norm utility vector $\mathbf{u}$ parameterized by
${\theta} = ({\boldsymbol{\mu}},\kappa)$ is given by:
\begin{equation}\label{eq04}
   \hspace{-2ex}f(\mathbf{u};\theta) =
   \frac{\exp(\kappa{\boldsymbol{\mu}}^T\mathbf{u})}
   {\int_{\mathbf{u}\in\mathcal{S}^{n-1}}
   \exp(\kappa{{\boldsymbol{\mu}}}^T\mathbf{u})d\mathbf{u}}
   \propto \exp(\kappa{{\boldsymbol{\mu}}}^T\mathbf{u}),
\end{equation}
where the unit vector $\boldsymbol{\mu} \in \mathbb{R}^n$ is the
mean direction and $\kappa \in \mathbb{R}_{\ge 0}$ is the
concentration parameter. When $\kappa = \infty$, the
distribution shrinks to a point mass on the unit hyper-sphere,
giving a deterministic $\mathbf{u}_k$ for all $k$.

\subsection{WDPCH-SU: Cognition Masking Algorithm and Privacy
Guarantee}
\label{sec:method}

The proposed cognition-hiding approach is summarized in
Algorithm~\ref{alg:wasserstein-privacy}, which operates under
a given privacy budget $\epsilon$. Instead of directly solving
the constrained optimization problem in~\eqref{eq:policy}, we
introduce a perturbed utility vector $\tilde{\mathbf{u}}_k$
obtained by adding zero-mean Laplace noise to the
vMF-distributed utility vector $\mathbf{u}_k$, with scale
parameter calibrated to the worst-case $\infty$-Wasserstein
distance $\Delta_W$.

\begin{algorithm}[t]
\caption{Wasserstein Distributional Privacy for Cognition
Hiding -- Static Utility (WDPCH-SU)}
\label{alg:wasserstein-privacy}
\begin{algorithmic}[1]
\Require Privacy budget $\epsilon$, mean direction
$\boldsymbol{\mu}_1$, concentration $\kappa_1$
\State Set $\Delta_W = 2$ (Lemma~\ref{lem:delta:W:2})

\ForEach{$k$}
\State Model utility:
$\mathbf{u}_k \sim \mathrm{vMF}(\boldsymbol{\mu}_1,\kappa_1)$
    \State Acquire scene perception vector
$\boldsymbol{\alpha}_k$ from the CR processor via inverse
filtering.
    \State Sample $\mathbf{z}_k = [z_{k,1},\ldots,z_{k,n}]$
    i.i.d.: $z_{k,i} \overset{\mathrm{i.i.d.}}{\sim}
    \mathcal{L}(0,\Delta_W/\epsilon)$.
    \State Compute perturbed utility:
    $\tilde{\mathbf{u}}_k = \mathbf{u}_k + \mathbf{z}_k$.
    \State Solve the constrained LP to obtain the
    privacy-preserving beam allocation:
    \begin{align}
	  &{\boldsymbol{\tilde{\beta}}^*_k }= \arg \max_{\boldsymbol{\beta}_k}\; (\boldsymbol{\tilde{u}}_k^T \boldsymbol{\beta}_k) \label{eqn:beta:privacy:preserving}\\ 
	  &\text{ s.t. }  \boldsymbol{{\alpha}_k}^T \boldsymbol{\beta}_k \leq 1, 0 \leq {{\beta}_{k,i}} \leq 1, \; \forall i, \nonumber
	\end{align}
\EndFor
\end{algorithmic}
\end{algorithm}
\begin{lemma}\label{lem:delta:W:2}
\begin{enumerate}
\item For $P_{\theta_i} = \mathrm{vMF}(\theta_i)$ and
$P_{\theta_j} = \mathrm{vMF}(\theta_j)$, the
$\infty$-Wasserstein distance
$W_\infty(P_{\theta_i},P_{\theta_j}) = 2$.
\item Let $\Psi$ be a collection of vMF distributions;
then $\Delta_W(\Psi) = 2$.
\end{enumerate}
\end{lemma}
The proof of Statement~1 follows from the triangle
inequality and the fact that samples from the vMF
distribution are always of unit norm. Statement~2 is a
direct consequence of Statement~1.

\begin{theorem}\label{thm:DistP}
Algorithm~\ref{alg:wasserstein-privacy} satisfies
$\epsilon$-DistP.
\end{theorem}
The proof follows from~\cite{Chen2023}, wherein the
Wasserstein mechanism of adding Laplace noise guarantees
$\epsilon$-DistP for $\mathbf{u}_k$. By the post-processing
immunity of distribution privacy~\cite{Dwork2006},
$\tilde{\boldsymbol{\beta}}_k^*$, being a deterministic
function of $\tilde{\mathbf{u}}_k$, also satisfies
$\epsilon$-DistP.

\subsection{Trade-off Between Privacy and Utility Loss}
\label{sec:tradeoff_static}
The introduction of noise in the vMF-distributed utility
vector affects radar performance. The trade-off is
characterised by the \emph{utility loss}
$|\Delta\mathbf{u}_k|$:
\begin{equation}
|\Delta\mathbf{u}_k| =
|\mathbf{u}_k^T(\boldsymbol{\beta}^*_k -
\tilde{\boldsymbol{\beta}}^*_k)|.
\label{eqn:Performance_loss}
\end{equation}

We first establish a general tail bound on the $\ell_1$-norm
of an i.i.d.\ Laplace noise vector, which we invoke both
here and in Theorem~\ref{thm:cost:loss} for the dynamic
setting.

\begin{lemma}[Laplace vector tail bound]
\label{lem:laplace:tail}
Let $\mathbf{z} = [z_1,\ldots,z_n]$ with
$z_i\overset{\mathrm{i.i.d.}}{\sim}\mathrm{Laplace}(0,\Delta_W/\epsilon)$.
Then, for any $\delta\in(0,1)$, with probability at least
$1-\delta$,
\begin{equation}
\|\mathbf{z}\|_1 \lesssim
\frac{2\Delta_W}{\epsilon}\sqrt{n\ln\!\left(\frac{1}{\delta}\right)}.
\label{eqn:laplace:tail:bound}
\end{equation}
\end{lemma}
\begin{proof}
Each component satisfies
$\mathbb{P}(|z_i|\geq\tau)\leq\exp(-\tau/(\Delta_W/\epsilon))$.
By the union bound,
$\mathbb{P}(\max_i|z_i|\geq\tau)\leq
n\exp(-\tau/(\Delta_W/\epsilon))$. Applying Bernstein's
inequality,
\[
\mathbb{P}\!\left(\|\mathbf{z}\|_1\geq t\right)\leq
\exp\!\left(-\frac{t^2/2}{2n(\Delta_W/\epsilon)^2+
\frac{\tau t}{3}}\right).
\]
For large $n$, this gives $\|\mathbf{z}\|_1\lesssim
\frac{2\Delta_W}{\epsilon}\sqrt{n\ln(1/\delta)}$ with
probability at least $1-\delta$.
\end{proof}

\begin{theorem}[Upper bound on utility loss]
\label{thm:utility:loss}
With probability at least $1 - 2\delta$, the utility loss
satisfies
\begin{equation}
|\Delta\mathbf{u}_k| \leq
\frac{2\Delta_W}{\epsilon}
\sqrt{n\ln\!\left(\frac{1}{\delta}\right)}.
\label{eqn:upperbound}
\end{equation}
\end{theorem}
\begin{proof}
Since $\tilde{\boldsymbol{\beta}}^*_k$ is optimal under
$\tilde{\mathbf{u}}_k$, we have $\tilde{\mathbf{u}}_k^T
\tilde{\boldsymbol{\beta}}^*_k \geq \tilde{\mathbf{u}}_k^T
\boldsymbol{\beta}^*_k$. Moreover, since
$\boldsymbol{\beta}^*_k$ is optimal under $\mathbf{u}_k$ over
the same feasible set (the constraints
in~\eqref{eq:LP-unperturbed} do not depend on $\mathbf{u}_k$),
$\mathbf{u}_k^T(\tilde{\boldsymbol{\beta}}^*_k-\boldsymbol{\beta}^*_k)
\leq 0$. Substituting $\tilde{\mathbf{u}}_k =
\mathbf{u}_k + \mathbf{z}_k$ into the first inequality and
combining with the second gives
$|\Delta\mathbf{u}_k| \leq |\mathbf{z}_k^T(
\tilde{\boldsymbol{\beta}}^*_k - \boldsymbol{\beta}^*_k)|
\leq \|\mathbf{z}_k\|_1
\|\tilde{\boldsymbol{\beta}}^*_k-\boldsymbol{\beta}^*_k\|_\infty
\leq \|\mathbf{z}_k\|_1$,
where the penultimate step follows from H\"{o}lder's
inequality and the last from
$\boldsymbol{\beta}^*_k,\tilde{\boldsymbol{\beta}}^*_k\in[0,1]^n$.
The bound then follows from Lemma~\ref{lem:laplace:tail}
applied to $\mathbf{z}_k$.
\end{proof}

\section{Masking the Dynamic Utility of a Cognitive Radar}
\label{sec:dynamic}

Section~\ref{sec:cogradar} addressed cognition masking in
the static setting, where, for example, the radar selects a
beam allocation at each instant based on its current belief.
Beyond this single-stage scenario, a multifunction cognitive
radar (MFR) may also operate over the long run,
continually switching between operating modes, such as
coarse scanning, fine scanning, coarse tracking, and fine
tracking, in order to manage multiple targets efficiently
under time-varying operating conditions. This switching
among operating modes, referred to as the radar's
\emph{sensing plan}, is governed by the long-run average
utility the radar accrues from operating in each mode.

\subsection{MDP Framework for Dynamic Utility Maximization}
\label{sec:dynamic_framework}

\begin{figure}[t]
    \centering
    \includegraphics[width=\linewidth]{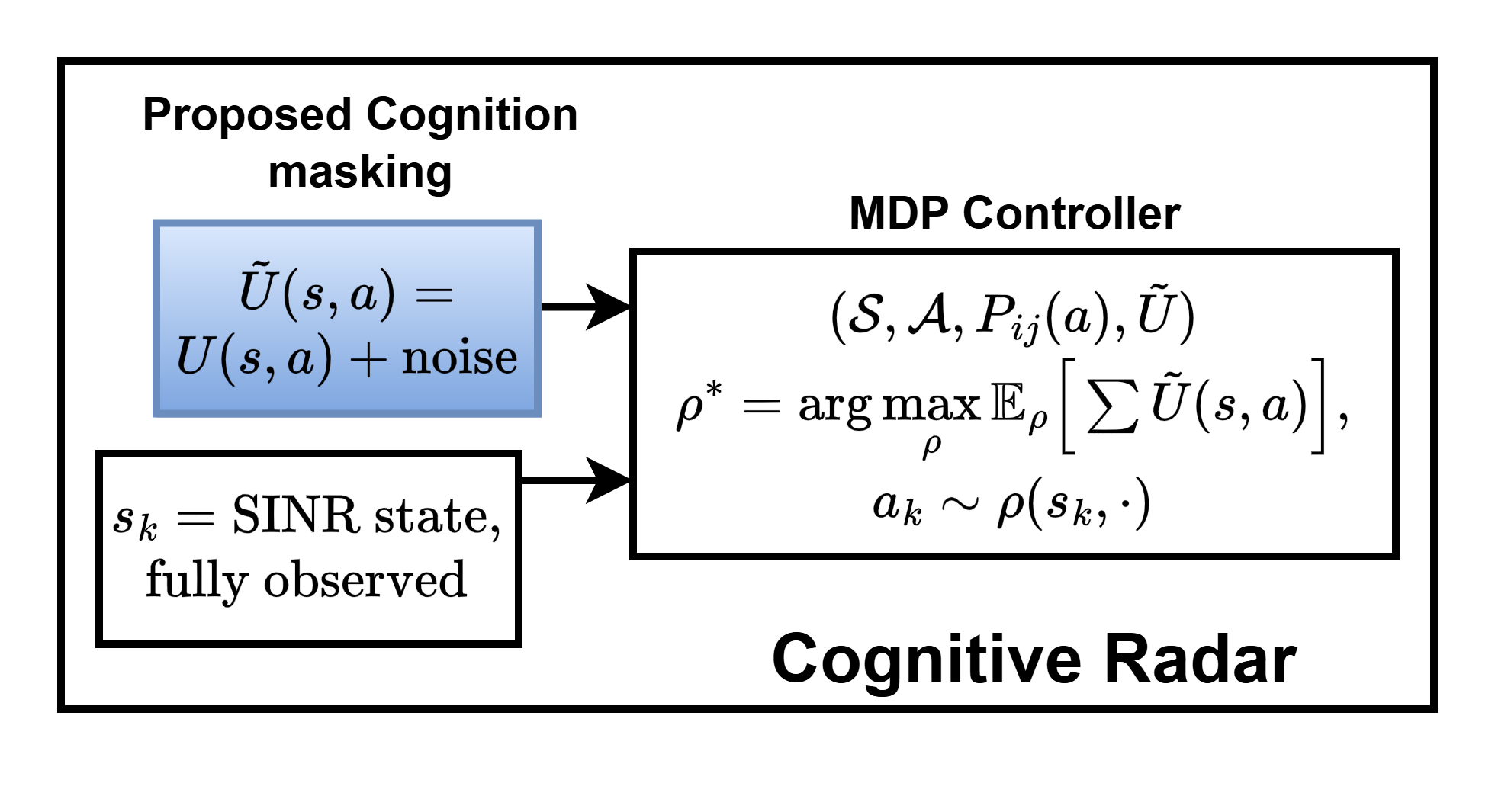}
  \caption{MDP controller for the dynamic setting under the
proposed cognition masking mechanism. The perturbed utility
$\tilde{U}(s,a) = U(s,a) + \text{noise}$, together with the
fully observed SINR state $s_k$, is supplied to the MDP
controller $(\mathcal{S}, \mathcal{A}, P_{ij}(a), \tilde{U})$,
which computes the privacy-preserving policy $\rho^*$ and
selects actions $a_k \sim \rho(s_k,\cdot)$, thereby masking
the true utility $U(s,a)$ from an adversary observing the
resulting trajectory $\{(s_k,a_k)\}_{k=1}^K$.}
    \label{fig:mdp_controller}
\end{figure}

We model the radar's mode-switching controller as an MDP
$(\mathcal{S}, \mathcal{A}, P_{ij}(a), U)$, comprising a
finite state space $\mathcal{S}$, a finite action space
$\mathcal{A}$, a transition probability matrix $P_{ij}(a)$
for $i,j \in \mathcal{S}$ and $a \in \mathcal{A}$, and a
state--action utility $U(s,a)$ that captures the operational
performance and resource efficiency of taking action $a$ in
state $s$. Fig.~\ref{fig:mdp_controller} illustrates this MDP
controller together with the proposed cognition-masking
mechanism, detailed in 	Sec.~\ref{sec:dynamic}-\ref{sec:wdpch-du}, which
perturbs $U(s,a)$ before policy computation.

As a motivating example, following~\cite{JainRadarConf2024,
JainTAES2025}, consider a scenario in which the radar's
operating condition is characterized by the instantaneous
received Signal-to-Interference-plus-Noise Ratio (SINR)
$\rho$. Since the radar operates over a fixed receive SINR
range, this range is uniformly discretised into
$|\mathcal{S}|$ bins, with each state $s \in \mathcal{S}$
corresponding to one such SINR interval, determined jointly
by several radar parameters (e.g., transmit power, PRF,
bandwidth, and dwell time). The radar's actions $a \in
\mathcal{A}$ correspond to its operating modes: fine
scanning, coarse scanning, fine tracking, and coarse
tracking, i.e., $\mathcal{A} = \{\text{FS}, \text{CS},
\text{FT}, \text{CT}\}$, in the canonical example of
	Sec.~\ref{sec:experiments}-\ref{sec:dynamic_results}-\ref{sec:Quantitative}, where scanning searches for new
targets and tracking maintains custody of detected ones, at
coarse or fine resolution.

The transition probability matrix $P_{ij}(a)$ governs how the radar's SINR
state evolves as a function of the action taken: intuitively,
targets that are harder to detect or track---e.g., those in
low-SINR regions---require the radar to spend more time in
those states, regardless of whether it is scanning or tracking.
More generally, the transition kernel embeds the coupling
between the radar's resource-allocation choices and the
physical propagation/target environment, rather than being
under the radar's direct, arbitrary control. While SINR-based
states and search/track actions serve as our running example,
the MDP formulation above applies generally to any state space
over which the radar's operating condition evolves and any set
of operating modes over which it switches.

Over the long run, the radar seeks an optimal
stationary policy $\rho^*:\mathcal{S}\times\mathcal{A}\to[0,1]$
that maximizes the long-run average expected utility
\begin{multline}
V^\rho(i) :=
\lim_{N\to\infty}\frac{1}{N+1}\mathbb{E}^{\rho}
\left[\sum_{k=0}^{N}U(s_k,a_k)\;\middle|\; s_0=i\right], \\
\quad i\in\mathcal{S},
\label{eq:radar_dynamic_obj}
\end{multline}
by solving
\begin{equation}
\rho^* = \arg\max_{\rho}\, V^\rho(i), \quad \forall\, i \in \mathcal{S},
\label{eq:radar_dynamic_policy}
\end{equation}
where $\mathbb{E}^{\rho}$ denotes the expectation over the
state--action trajectory induced by the stationary policy
$\rho$ and the transition dynamics $P_{ij}(a)$. The optimal
policy $\rho^*$, defined in~\eqref{eq:radar_dynamic_policy},
can be equivalently characterised by the state-action
occupancy measure $\rho(s,a) = \Pr(s_k=s, a_k=a)$ under $\rho^*$, which satisfies the following linear
program~\cite{Puterman1994,Altman1999}:
\begin{equation}
\rho^* = \arg\max_{\rho}\;
\sum_{i\in\mathcal{S}}\sum_{a\in\mathcal{A}}
U(i,a)\,\rho(i,a)
\label{eq:radar_dynamic_lp}
\end{equation}
\begin{align*}
\text{s.t.} \quad & \rho(i,a) \ge 0, \quad \forall\, i \in \mathcal{S},\; a \in \mathcal{A}, \\
& \sum_{a} \rho(j,a) = \sum_{i}\sum_{a} P_{ij}(a)\,\rho(i,a), \quad \forall\, j \in \mathcal{S}, \\
& \sum_{i}\sum_{a} \rho(i,a) = 1.
\end{align*}

\subsection{Adversarial Inverse Learning}

In this dynamic setting, the threat from adversarial
inference is qualitatively different from the static case of
Sec.~\ref{sec:cogradar}-\ref{sec:static_framework}. Here, the adversary passively
observes the radar's state--action trajectory as it switches
between operating modes over the long run, thereby
accumulating the dataset $\mathcal{D} =
\{(s_k,a_k)\}_{k=1}^K$, and seeks to infer the radar's latent
utility $U$ from $\mathcal{D}$.

Since the optimal MDP policy $\rho^*$
in~\eqref{eq:radar_dynamic_policy} is
deterministic~\cite{JainRadarConf2024,JainTAES2025}, the
radar's sensing plan can, for example, be estimated using
maximum-likelihood estimation, with estimation accuracy
increasing as $K \to \infty$. Prior work has addressed this
vulnerability by perturbing the sensing plan to reduce the
adversary's Fisher information~\cite{JainRadarConf2024,
JainTAES2025}; however, such formulations are non-convex and
do not provide formal statistical guarantees. Motivated by
this gap, we extend the distributional privacy framework of
Sec.~\ref{sec:cogradar} to the dynamic setting, so as to
provide a formal $\epsilon$-DistP guarantee against this class of sustained inference of the radar's sensing plan.

\subsection{Stochastic Utility Model and Policy Optimization}
\label{sec:stochastic_util_dynamic}

Building on the MDP framework of 	Sec.~\ref{sec:dynamic}-\ref{sec:dynamic_framework}
and the adversarial inverse learning problem of the previous
subsection, we now specify the stochastic prior governing the
state--action utility $U$, which constitutes the private
information the radar seeks to protect. As in the static
setting of Sec.~\ref{sec:cogradar}-\ref{sec:stochastic_util}, we model the
vectorised utility table
$U\in\mathbb{R}^{|\mathcal{S}||\mathcal{A}|}$ as
$U\sim\mathrm{vMF}(\boldsymbol{\mu},\kappa)$, where $\theta =
(\boldsymbol{\mu},\kappa)$ is the latent cognitive parameter
constituting the private information to be protected. The
entry $U(s,a)$ gives the utility at state $s$ under action
$a$. Since $U$ enters~\eqref{eq:radar_dynamic_lp} linearly,
scaling $U$ by a positive constant does not affect the
optimal policy, so $U$ can be restricted to the unit
hypersphere $\mathcal{S}^{n-1}$ where
$n:=|\mathcal{S}||\mathcal{A}|$, following~\cite{Mardia2009}.

\subsection{WDPCH-DU: Cognition Masking Algorithm and Privacy
Guarantee}
\label{sec:wdpch-du}

To protect the radar's latent utility against the
inverse-learning adversary, regardless of the specific
inference method, we perturb the
utility prior to policy computation by injecting Laplace
noise independently at every state--action pair:
\begin{equation}
\begin{aligned}
\tilde{U}(s,a) &= U(s,a) + z(s,a), \\
z(s,a) &\overset{\mathrm{i.i.d.}}{\sim}
\mathrm{Laplace}\!\left(0,\frac{\Delta_W}{\epsilon}\right),
\quad \forall\,(s,a) \in \mathcal{S}\times\mathcal{A},
\end{aligned}
\end{equation}
where $\Delta_W=2$ (by Lemma~\ref{lem:delta:W:2}) is the
Wasserstein sensitivity and $\epsilon$ is the privacy
budget. The radar policy is then computed using $\tilde{U}$
instead of $U$, masking the true utility distribution. The
proposed approach is summarised in
Algorithm~\ref{alg:cost-privacy}.

\begin{algorithm}[t]
\caption{Wasserstein Distributional Privacy for Cognition
Hiding -- Dynamic Utility (WDPCH-DU)}
\label{alg:cost-privacy}
\begin{algorithmic}[1]
\Require Privacy budget $\epsilon$, mean direction
$\boldsymbol{\mu}_2$, concentration $\kappa_2$

\State $\Delta_W \gets 2$ (Lemma~\ref{lem:delta:W:2})
\State $U(s,a) \sim \mathrm{vMF}(\boldsymbol{\mu}_2,\kappa_2)$
\State Sample $z(s,a) \overset{\mathrm{i.i.d.}}{\sim}
\mathrm{Laplace}(0,\Delta_W/\epsilon)$, $\forall\, (s,a) \in
\mathcal{S}\times\mathcal{A}$
\State $\tilde{U}(s,a) \gets U(s,a) + z(s,a)$,
$\forall\, (s,a) \in \mathcal{S}\times\mathcal{A}$
\State $\rho \gets$ solution of~\eqref{eq:radar_dynamic_lp}
with $\tilde{U}$ in place of $U$
\Ensure Privacy-preserving policy $\rho$
\end{algorithmic}
\end{algorithm}

\begin{theorem}
\label{thm:DistP:DU}
Algorithm~\ref{alg:cost-privacy} satisfies $\epsilon$-DistP.
\end{theorem}
The proof follows analogously to Theorem~\ref{thm:DistP},
with $U(s,a)$ replacing $\mathbf{u}_k$ and $\rho$ replacing
$\tilde{\boldsymbol{\beta}}_k^*$: the Wasserstein mechanism
in~\cite{Chen2023} guarantees $\epsilon$-DistP for the
privatised utility, and post-processing
immunity~\cite{Dwork2006} ensures that $\rho$, as a
deterministic function of $\tilde{U}$, also satisfies
$\epsilon$-DistP.

\subsection{Trade-off Between Privacy and Utility}
\label{sec:tradeoff_dynamic}

The introduction of noise in the vMF-distributed utility
affects the radar policy. The trade-off is characterised by
the utility deviation $|U^\top(\tilde{\rho}-\rho^*)|$.

\begin{theorem}[Upper bound on utility deviation]
\label{thm:cost:loss}
With probability at least $1-2\delta$, the utility
deviation satisfies
\begin{equation}
\left|U^\top(\tilde{\rho}-\rho^*)\right|
\leq\frac{2\Delta_W}{\epsilon}
\sqrt{n\ln\!\left(\frac{1}{\delta}\right)},
\end{equation}
where $n=|\mathcal{S}||\mathcal{A}|$ is the dimension of
the utility vector.
\end{theorem}
\begin{proof}
Since $\tilde{\rho}$ is optimal under $\tilde{U}=U+z$, we
have $(U+z)^\top\tilde{\rho}\geq(U+z)^\top\rho^*$. Moreover,
since $\rho^*$ is optimal under $U$ over the same feasible
set (the constraints in~\eqref{eq:radar_dynamic_lp} do not
depend on $U$), $U^\top\rho^*\geq U^\top\tilde{\rho}$, so
$U^\top(\tilde{\rho}-\rho^*)\leq 0$. Combining both
inequalities and rearranging gives
$|U^\top(\tilde{\rho}-\rho^*)|\leq
|z^\top(\tilde{\rho}-\rho^*)|\leq\|z\|_1\|\tilde{\rho}-\rho^*\|_\infty
\leq \|z\|_1$,
where the penultimate step follows from H\"{o}lder's
inequality and the last from
$\tilde{\rho},\rho^*$ being probability distributions over
state--action pairs ($\sum_{s,a}\rho(s,a)=1$,
$\rho(s,a)\geq0$). The bound then follows from
Lemma~\ref{lem:laplace:tail}  (Sec.~\ref{sec:cogradar}-\ref{sec:tradeoff_static}) applied
to $z$.
\end{proof}
\begin{remark}
The bound in Theorem~\ref{thm:cost:loss} takes the same
functional form as the bound in
Theorem~\ref{thm:utility:loss} because both settings inject
i.i.d.\ Laplace noise of the same scale into a
unit-sphere-supported distribution; the ambient dimension
$n=|\mathcal{S}||\mathcal{A}|$ (dynamic) matches the role of
$n$ (number of targets) in the static setting.
\end{remark}

Together, WDPCH-SU and WDPCH-DU mask the radar's cognitive
strategy under static and dynamic paradigms, respectively,
with formal $\epsilon$-DistP guarantees established for
each. We now validate both mechanisms numerically.

\section{Numerical Results}\label{sec:experiments}

\begin{figure}[t]
    \centering
    \includegraphics[width=\linewidth]{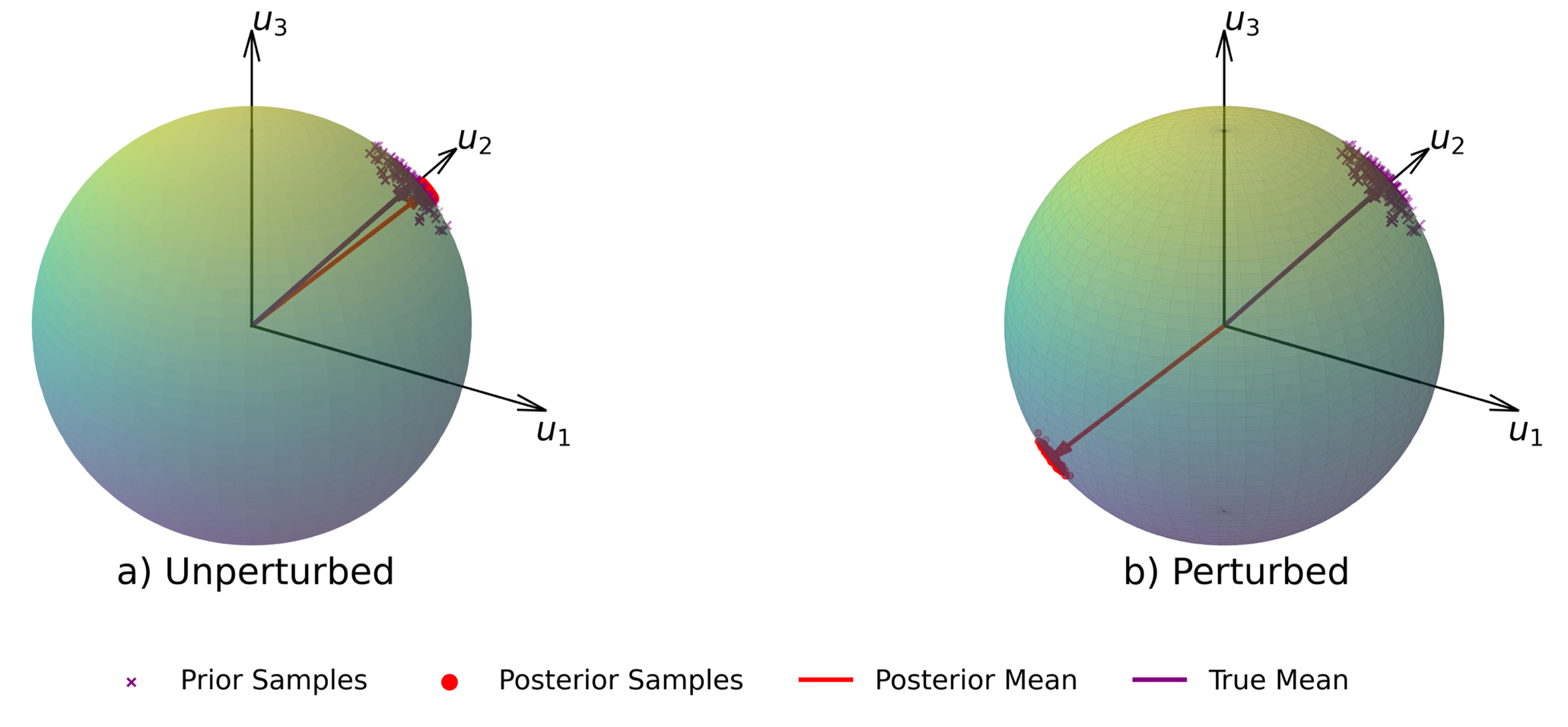}
    \caption{Comparison of adversarial inference results with
    and without the privacy-preserving mechanism. When
    Algorithm~\ref{alg:wasserstein-privacy} is applied
    (Fig.~(b)), the adversary's MH algorithm fails to recover
    the true utility distribution. Without the privacy
    mechanism (Fig.~(a)), the adversary successfully learns
    the utility distribution, highlighting the vulnerability
    of unprotected cognitive radar strategies.}
    \label{fig:quan}
\end{figure}

\subsection{Masking the Static Utility of CR}
\label{sec:static_results}
\subsubsection{Qualitative Analysis}\label{sec:Qualitative}

We consider a cognitive radar (CR) operating in a
multi-target tracking environment, where the radar
dynamically allocates its beam among $n$ targets to maximise
the average tracking accuracy, as illustrated in
Fig.~\ref{Fig:radar-target}. An adversary performs targeted
maneuvers and observes the radar's beam allocation decisions
$\boldsymbol{\beta}_k^*$ with the objective of learning the
underlying allocation strategy. The CR employs
Algorithm~\ref{alg:wasserstein-privacy} both to optimise
tracking performance and to obfuscate its true allocation
strategy.

\textit{Experimental setup}: We assume:
\begin{inparaenum}[(i)]
\item number of targets $n = 3$;
\item at each time $k$, $\mathbf{u}_k$ is sampled from a
vMF distribution with mean $\boldsymbol{\mu}_1 =
[0.57, 0.47, 0.67]$ and concentration $\kappa_1 = 60$,
corresponding to a highly concentrated utility distribution
(close to deterministic); and
\item number of observations $N = 10000$, with dataset
$\mathcal{D} = \{(\boldsymbol{\alpha}_k,
\boldsymbol{\beta}_k^*)\}_{k=1}^N$. Details on dataset
generation are given
in~\cite{Krishnamurthy2020,Krishna2020Rd,Bell2015}.
\end{inparaenum}

The adversary observes $\mathcal{D}$ and estimates $\theta$
via the Metropolis-Hastings (MH) algorithm
(Appendix~\ref{alg:MH}), computing the posterior
$P(\boldsymbol{\mu}_1 \mid \mathcal{D})$ (assuming $\kappa_1$
known), where $\boldsymbol{\mu}_1$ and $\kappa_1$ are the
static-setting vMF parameters as in
Algorithm~\ref{alg:wasserstein-privacy}. The results are
illustrated in Figure~\ref{fig:quan}, where \emph{prior
samples} (blue) refer to samples from
$\mathrm{vMF}(\boldsymbol{\mu}_1, \kappa_1)$ and posterior
samples are shown in red.

\paragraph{Case 1: Without distribution privacy}
As shown in Figure~\ref{fig:quan}a, the MH algorithm
converges to the true utility distribution, highlighting the
vulnerability of the beam allocation strategy. The cosine
similarity $\mathrm{CS}(\boldsymbol{\mu}_1,
\hat{\boldsymbol{\mu}}_1) = {\boldsymbol{\mu}_1^T
\hat{\boldsymbol{\mu}}_1}/({\|\boldsymbol{\mu}_1\|
\|\hat{\boldsymbol{\mu}}_1\|})$ equals $0.990$, indicating
near-perfect adversarial recovery.

\paragraph{Case 2: With
Algorithm~\ref{alg:wasserstein-privacy}}
The CR computes the \emph{perturbed beam allocation}
$\tilde{\boldsymbol{\beta}}_k^*$. For $\delta = 0.01$ and
$\epsilon = 0.1$, Figure~\ref{fig:quan}b shows the posterior
utility distribution. The adversary fails to infer the true
utility distribution $\mathrm{vMF}(\boldsymbol{\mu}_1,
\kappa_1)$, as demonstrated by the cosine similarity of
$-0.99$ between $\boldsymbol{\mu}_1$ and the adversary's
estimate $\hat{\boldsymbol{\mu}}_1$. Hence,
Algorithm~\ref{alg:wasserstein-privacy} provides a strong
privacy-preserving framework against adversarial inference
attacks.

\subsubsection{Quantitative Analysis}

We benchmark Algorithm~\ref{alg:wasserstein-privacy} against
the state-of-the-art method in~\cite{Pattanayak2023}.

\textit{Existing methodology~\cite{Pattanayak2023}}:
\cite{Pattanayak2023} considers a general (possibly
nonlinear) utility function $u(\boldsymbol{\beta}_k)$, of
which our linear utility $\mathbf{u}_k^T\boldsymbol{\beta}_k$
in~\eqref{eq:policy} is a special case. Under the assumption
that $u(\boldsymbol{\beta}_k)$ is continuous, monotone and
concave, \cite{Pattanayak2023} uses \emph{revealed
preference}~\cite{Varian2012} to derive a necessary and
sufficient set of linear inequalities (referred to as
\emph{Afriat's inequalities})~\cite{Afriat1967}, denoted
$\mathcal{A}_u(\mathcal{D})$. For the linear utility
function considered here, \eqref{eqn:afriat:ineqs} provides
the equivalent Afriat's inequalities via duality of the LP
in~\eqref{eq:policy}. Define
$\mathbf{V}^{(\mathcal{D})} = [\mathbf{V}^{(1)},
\ldots, \mathbf{V}^{(N)}]^T$ and
$\mathbf{w}^{(\mathcal{D})} = [\mathbf{w}^{(1)},
\ldots, \mathbf{w}^{(N)}]^T$; then
$\mathcal{A}_u(\mathcal{D}) = \left\{\mathbf{V}^{(k)}
\mathbf{u}_k \leq \mathbf{w}^{(k)}\right\}_{k=1}^N$.
\begin{subequations}\label{eqn:afriat:ineqs}
\begin{align}
&\text{When } \boldsymbol{\alpha}_k^T\boldsymbol{\beta}_k =
1: -\min_{i:\boldsymbol{\beta}_{k}^*(i)>0}
\left(\frac{\mathbf{u}_k(i)}{\boldsymbol{\alpha}_k(i)}
\right) < 0, \label{eqn:afriat:ineqs:a}\\
&\max_{i:\boldsymbol{\beta}_{k}^*(i)=0}
\left(\frac{\mathbf{u}_k(i)}{\boldsymbol{\alpha}_k(i)}
\right) - \min_{i:\boldsymbol{\beta}_{k}^*(i)>0}
\left(\frac{\mathbf{u}_k(i)}{\boldsymbol{\alpha}_k(i)}
\right) \leq 0, \label{eqn:afriat:ineqs:b}\\
&\text{when } \boldsymbol{\alpha}_k^T\boldsymbol{\beta}_k 
1: \max_{i:\boldsymbol{\beta}_{k}^*(i)=0,
\mathbf{u}_k(i)\neq 0}
\left(\frac{\mathbf{u}_k(i)}{\boldsymbol{\alpha}_k(i)}
\right) \leq 0. \label{eqn:afriat:ineqs:c}
\end{align}
\end{subequations}
The feasibility margin $\mathcal{M}(\mathcal{D})$ is defined
as:
\begin{equation}
\mathcal{M}(\mathcal{D}) = \min_{\delta_A \geq 0} \left\{
\delta_A,\; \mathcal{A}_u(\mathcal{D}) + \delta_A\mathbf{1}
\geq 0 \right\},
\label{eqn:feaibility:margin}
\end{equation}
where $\mathbf{1}$ is a column vector of all ones and
$\delta_A \geq 0$ is the minimum perturbation required for
the Afriat inequalities to fail. The perturbed response
$\{\tilde{\boldsymbol{\beta}}_k^*\}_{k=1}^N$ is obtained
by solving:
\begin{align}
&\left\{\tilde{\boldsymbol{\beta}}_k^*\right\}_{k=1}^{N} =
\underset{\{\boldsymbol{\beta}_k\}_{k=1}^{N}}{\arg\min}
\quad \sum_{k=1}^{N} \left(u(\boldsymbol{\beta}_k^*) -
u(\boldsymbol{\beta}_k)\right) \label{eqn:kunal:opt}\\
&\text{s.t.}\quad \boldsymbol{\beta}_k \geq 0, \quad
\boldsymbol{\alpha}_k^T\boldsymbol{\beta}_k \leq 1, \quad
\mathcal{M}_u(\tilde{\mathcal{D}}) \leq
(1-\eta)\mathcal{M}_u(\mathcal{D}), \nonumber
\end{align}
where $\tilde{\mathcal{D}} = \{\boldsymbol{\alpha}_k,
\tilde{\boldsymbol{\beta}}_k^*\}_{k=1}^N$ and $\eta \in
[0,1]$ controls the extent of cognition masking. For
comparison purposes, $\epsilon = 1-\eta$, establishing a
direct relationship between the privacy budget and the
masking strength.

\begin{figure}[t]
    \centering
    \subfloat[Algorithm~\ref{alg:wasserstein-privacy}]{%
        \includegraphics[width=0.48\linewidth]{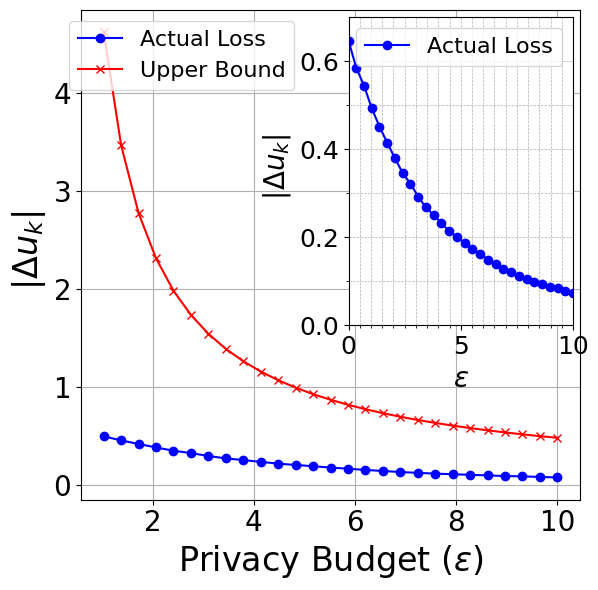}
        \label{fig:privacy_loss}}
    \hfill
    \subfloat[Existing methodology~\cite{Pattanayak2023}]{%
        \includegraphics[width=0.48\linewidth]{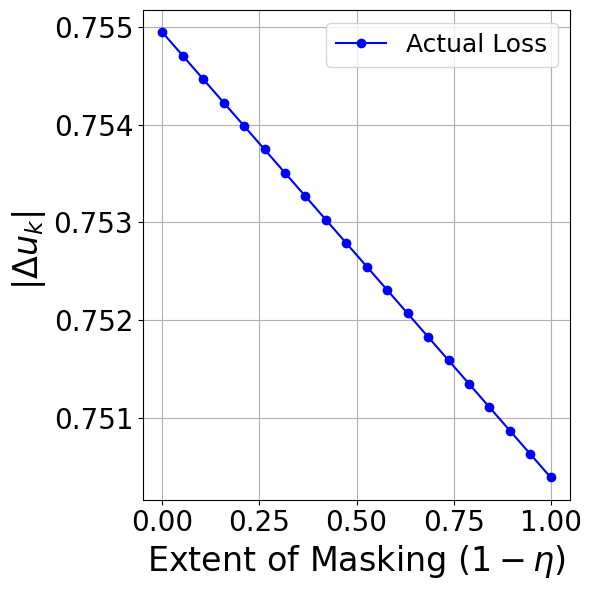}
        \label{fig:masking_loss}}
    \caption{Comparison of performance loss between
    Algorithm~\ref{alg:wasserstein-privacy} and the existing
    approach in~\cite{Pattanayak2023}. At perfect privacy,
    the proposed algorithm incurs lower utility loss. The loss
    reduces as $1/\epsilon$ (proposed) vs.\ linearly
    (existing).}
    \label{fig:performance_comparison}
\end{figure}

The experimental setup is as in
	Sec.~\ref{sec:experiments}-\ref{sec:static_results}-\ref{sec:Qualitative}. Performance is benchmarked in
terms of the utility loss in~\eqref{eqn:Performance_loss}.
Figure~\ref{fig:privacy_loss} shows the utility loss for
$\epsilon \in [10^{-4}, 10]$, along with the upper bound
from Theorem~\ref{thm:utility:loss}. The upper bound is
tight only for large $\epsilon$ due to the union bound and
Bernstein inequality. Figure~\ref{fig:masking_loss} shows
the utility loss versus $1-\eta$ for the existing approach.

From Figures~\ref{fig:privacy_loss}
and~\ref{fig:masking_loss}, we observe:
\begin{enumerate}
\item As $\epsilon$ (or $1-\eta$) increases (reduced
privacy), the empirical utility loss $|\Delta\mathbf{u}_k|$
decreases.
\item At perfect privacy ($\epsilon = 1-\eta = 0$), the
utility loss of Algorithm~\ref{alg:wasserstein-privacy} is
$0.646$, lower than the existing
methodology~\cite{Pattanayak2023} ($0.755$).
\item The utility loss in Figure~\ref{fig:privacy_loss}
reduces approximately as $1/\epsilon$, consistent with the
upper bound from Theorem~\ref{thm:utility:loss}. The
utility loss in Figure~\ref{fig:masking_loss} reduces
linearly, as expected from the linear constraint
in~\eqref{eqn:kunal:opt}.
\end{enumerate}

Hence, Algorithm~\ref{alg:wasserstein-privacy} offers a
better trade-off between privacy and utility loss.

\textit{Disadvantages of the existing approach
in~\cite{Pattanayak2023}}:
\begin{inparaenum}[(i)]
\item The optimization problem in~\eqref{eqn:kunal:opt} is
a batch problem, unsuitable for online electronic warfare
scenarios, whereas
Algorithm~\ref{alg:wasserstein-privacy} operates online
with minimal controller modifications.
\item There is no formal privacy guarantee for the approach
in~\cite{Pattanayak2023}: the feasibility margin is reduced
against revealed-preference learners, but no guarantee
exists against other inference strategies.
\end{inparaenum}

\subsection{Masking the Dynamic Utility of CR}
\label{sec:dynamic_results}
We consider a cognitive radar operating in a dynamic
decision-making environment, switching between operating
modes---such as scanning and tracking---according to a
stationary policy $\rho$, as described in
Sec.~\ref{sec:dynamic}-\ref{sec:dynamic_framework}. The radar employs
Algorithm~\ref{alg:cost-privacy}, which perturbs the
vMF-distributed utility $U$ prior to policy computation,
obfuscating the radar's latent utility from the adversary. To
evaluate this concealment against an adversary performing
inverse learning, we consider PolicyWalk~\cite{Ramachandran2007}
(Algorithm~\ref{alg:policywalk}) applied to demonstrations
generated both with and without the utility perturbation, and
examine whether the radar's true preference ordering between
operating states can still be recovered from the resulting
posterior.

\subsubsection{Qualitative Analysis}

\begin{figure}[t]
    \centering
    \includegraphics[width=0.9\linewidth]{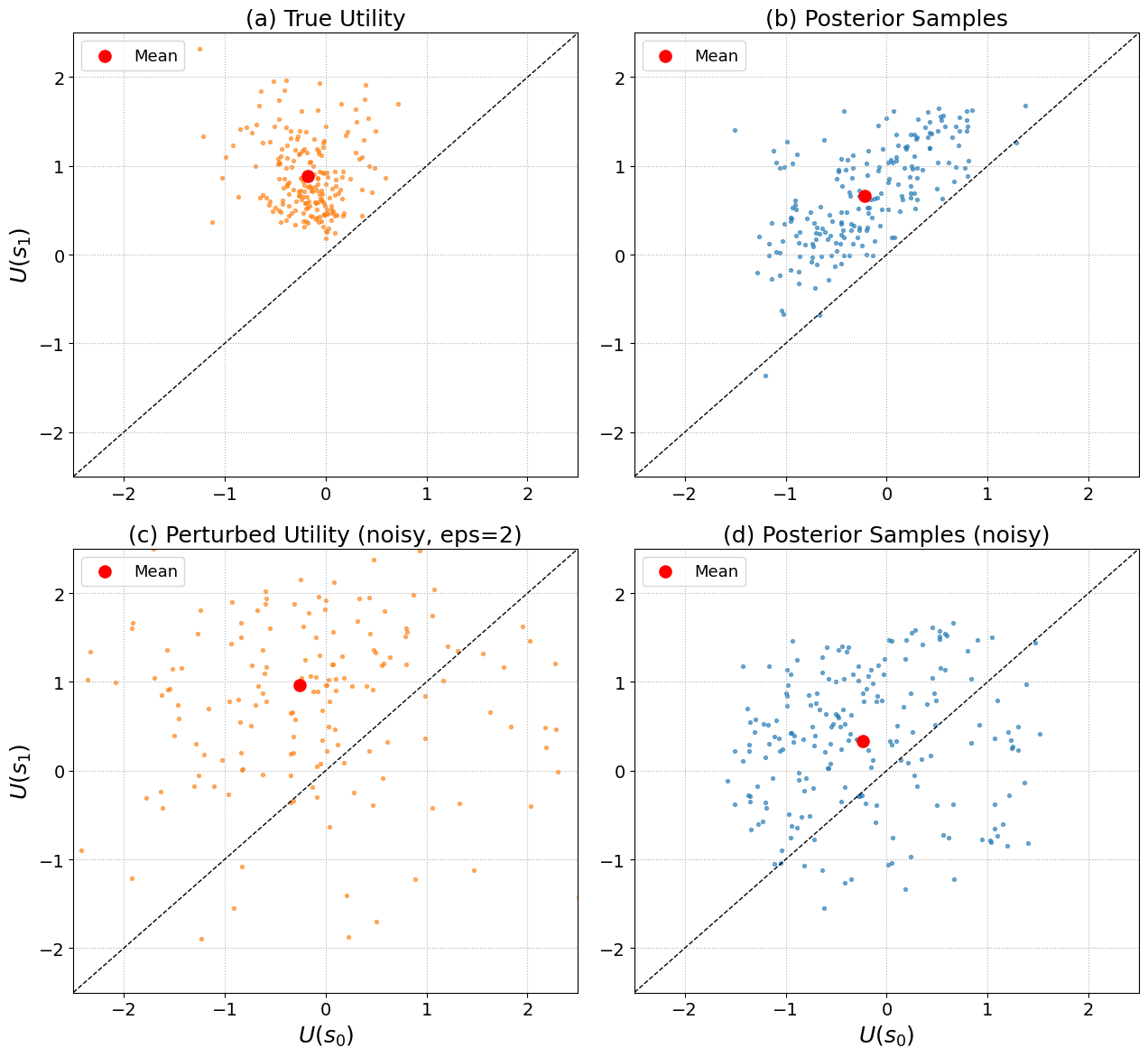}
    \caption{Effect of Algorithm~\ref{alg:cost-privacy} on the
    true utility $U$ and on PolicyWalk's recovered posterior,
    projected onto $(U(S_0),U(S_1))$, over 200 Monte Carlo
    trials. \textbf{Dashed diagonal:} the decision boundary
    $U(S_0)=U(S_1)$ — points above satisfy the true preference
    $U(S_1)>U(S_0)$, points on or below it do not. (a)~True
    utility without perturbation, showing the true preference
    $U(S_1)>U(S_0)$. (b)~Posterior of the true utility recovered
    via PolicyWalk, indicating the same preference as in (a).
    (c)~True utility perturbed by
    Algorithm~\ref{alg:cost-privacy}. (d)~Posterior of the
    perturbed utility recovered via PolicyWalk. Perturbing $U$
    disperses the scatter and pulls the red mean toward the
    diagonal in (d), showing that the preference ordering is no
    longer reliably recovered.}
    \label{fig:cost_comparison}
\end{figure}

We use the same four-state, two-action MDP
$(\mathcal{S},\mathcal{A},P_{ij}(a))$, $|\mathcal{S}|=4$,
$|\mathcal{A}|=2$, with transition probability matrix
$P_{ij}(a)$ adopted from the illustrative MDP
in~\cite{Ramachandran2007}, where $S_1$ is a higher-SINR radar
state than $S_0$. For this case study, we consider the utility
to depend only on the state, $U(s)$, rather than on the
state--action pair as in
Sec.~\ref{sec:dynamic}-\ref{sec:dynamic_framework}. As
in the static setting, $U$ is drawn from the stochastic linear
utility model of
Sec.~\ref{sec:cogradar}-\ref{sec:stochastic_util}, $U \sim
\mathrm{vMF}(\boldsymbol{\mu},\kappa)$ with $\boldsymbol{\mu_2}
\propto [-0.1,0.5,0.02,0.02]$ and $\kappa_2 = 12$, which
concentrates mass on utility vectors satisfying $U(S_1) >
U(S_0)$ — the radar's physically motivated preference for the
higher-SINR state $S_1$, whose stronger detection and tracking
performance makes it the dominant operating choice whenever
reachable. Recovering this ordering is the ``cognition'' an
adversary performing inverse learning seeks to expose: with
PolicyWalk, the recovered posterior reflects a similar
preference $U(S_1) > U(S_0)$ as the true utility, whereas under
Algorithm~\ref{alg:cost-privacy}, this preference is no longer
exactly identifiable.

We evaluate the proposed mechanism by comparing two cases. In
Case~1 (without distribution privacy), the radar computes its
policy directly from the true utility $U$, generates
demonstrations accordingly, and the adversary recovers a
posterior over $U$ using PolicyWalk
(Algorithm~\ref{alg:policywalk}), following the likelihood,
prior, and sampling procedure of~\cite{Ramachandran2007}. In
Case~2 (with Algorithm~\ref{alg:cost-privacy}), the true
utility is instead perturbed as $\tilde{U} = U + \mathbf{z}$,
with $\mathbf{z}$ drawn i.i.d.\ Laplace at scale $b =
\Delta_W/\epsilon$ for privacy budget $\epsilon = 2$; the radar re-solves its
policy under $\tilde{U}$, generates a fresh set of
demonstrations, and the adversary re-runs PolicyWalk
identically to recover a second posterior. Comparing the two
recovered posteriors thus directly reveals how effectively
Algorithm~\ref{alg:cost-privacy} conceals the radar's true
preference ordering. We repeat this over 200 Monte Carlo
trials and report all four resulting point clouds together
with their means.

\paragraph{Case 1: Without distribution privacy: Panels (a)--(b)}
The sampled true utilities (a) lie almost entirely above the
diagonal, confirming $U(S_1) > U(S_0)$ with high consistency,
and the mean (red marker) sits well above it. The recovered
posterior (b) stays concentrated above the diagonal with its
mean close to that in (a), correctly recovering the radar's
preference ordering in nearly all trials.

\paragraph{Case 2: With Algorithm~\ref{alg:cost-privacy}: Panels (c)--(d)}
Once $U$ is perturbed, the resulting $\tilde{U}$ (c) is
visibly more dispersed, with a non-negligible fraction of
samples now falling on or below the diagonal — the ordering
degrades at the level of the utility itself, before any
inference takes place. The corresponding posterior (d) straddles
the diagonal, with its mean shifting noticeably toward it
relative to (b), showing that the preference ordering is
recovered far less reliably. This degradation follows from
Algorithm~\ref{alg:cost-privacy} propagating through the radar's
policy, its demonstrations, and the PolicyWalk posterior — confirming
that privatising the utility alone is sufficient to conceal the
radar's SINR-driven preference from an adversary.

\subsubsection{Quantitative Analysis}\label{sec:Quantitative}

We consider an MDP with $|\mathcal{S}| = 10$ states and
$|\mathcal{A}| = 4$ actions, giving $n =
|\mathcal{S}||\mathcal{A}| = 40$, as in
	Sec.~\ref{sec:dynamic}-\ref{sec:dynamic_framework}. The SINR ranges from
$4$~dB to $44$~dB. The radar selects among
$\mathcal{A} = \{\text{FS}, \text{CS}, \text{FT}, \text{CT}\}$. The
utility function and transition probability matrix are
adopted from~\cite{JainTAES2025}.

\textit{Existing methodology}: We benchmark
Algorithm~\ref{alg:cost-privacy} against two prior approaches.~\cite{JainRadarConf2024,JainTAES2025} models the
multi-function radar's controller as an MDP with
state--action cost $c(i,u)$, and formulates the unperturbed
sensing plan as the cost-minimizing occupancy-measure LP
\begin{equation}
\rho_0 = \arg\min_{\rho}\; \textstyle\sum_{i,u} c(i,u)\,\rho(i,u),
\end{equation}
subject to the same constraints as in~\eqref{eq:radar_dynamic_lp};
the resulting optimal policy $\rho_0$ is deterministic.

To mask the sensing plan, \cite{JainTAES2025} perturbs $\rho$
away from $\rho_0$ while penalizing the determinant of the
adversary's Fisher Information Matrix (FIM), $\det F(\rho)$, as
follows:
\begin{multline}
\min_{\rho}\; \Big(\textstyle\sum_{i,u} c(i,u)\big(\rho(i,u)-\rho_0(i,u)\big)\Big)^{2} \\
+ \; \gamma \log\Big(\textstyle\prod_{i\in X}\sum_{u\in U}\rho(i,u)\Big),
\end{multline}
where $\gamma \geq 0$ trades off cost perturbation against
Fisher information reduction, subject to the same constraints
again as in~\eqref{eq:radar_dynamic_lp}. Although $\det
F(\rho)$ admits a closed-form expression, it is nonlinear in
$\rho$, so the resulting optimisation is non-convex, and
provides no formal statistical guarantee, such as
$\epsilon$-DistP, on what the adversary can infer about the
radar's sensing plan.

In addition, prior work~\cite{GohariCDC2020} takes a
structurally different approach: rather than perturbing the
policy or the utility, it perturbs the transition probability
matrix $P_{ij}(a)$ of the MDP directly, prior to policy
synthesis, via a Dirichlet mechanism with concentration
parameter $k>0$, and synthesises the policy for the resulting
privatised MDP. We investigate this approach in the cognitive
radar setting using the same MDP formulation as in
Sec.~\ref{sec:dynamic}-\ref{sec:dynamic_framework}. Smaller $k$ increases
randomness and hence strengthens privacy, at the cost of a
larger deviation of the resulting policy's value from that
under the true transition probabilities -- termed the ``cost
of privacy'' in \cite{GohariCDC2020}; larger $k$ weakens
privacy but reduces this cost. We vary $k \in [10^2, 10^4]$.

\textit{Quantitative Comparison}:
As shown in Fig.~\ref{fig:privacy_framework}(a), adversarial
Fisher information decreases as privacy increases across all
three approaches. For the proposed method, the decrease is
gradual; for the $\gamma$-based approach~\cite{JainTAES2025},
it is approximately constant up to a threshold and then
sharply declines. The proposed method achieves the largest
overall reduction in adversarial Fisher information among the
three approaches, outperforming both existing baselines.
Moreover, the existing methods require either explicit Fisher
information constraints (leading to non-convex optimisation,
as in~\cite{JainTAES2025}) or perturbation of the transition
probability matrix $P_{ij}(a)$, as in~\cite{GohariCDC2020}.
Transition perturbation may not always be appropriate for
cognitive radar systems, since the transition dynamics depend
on external factors such as target motion and channel
conditions, as well as the already-designed radar system.

\begin{figure}[t]
    \centering
    \begin{minipage}{0.4932\linewidth}
        \centering
        \includegraphics[width=\linewidth]{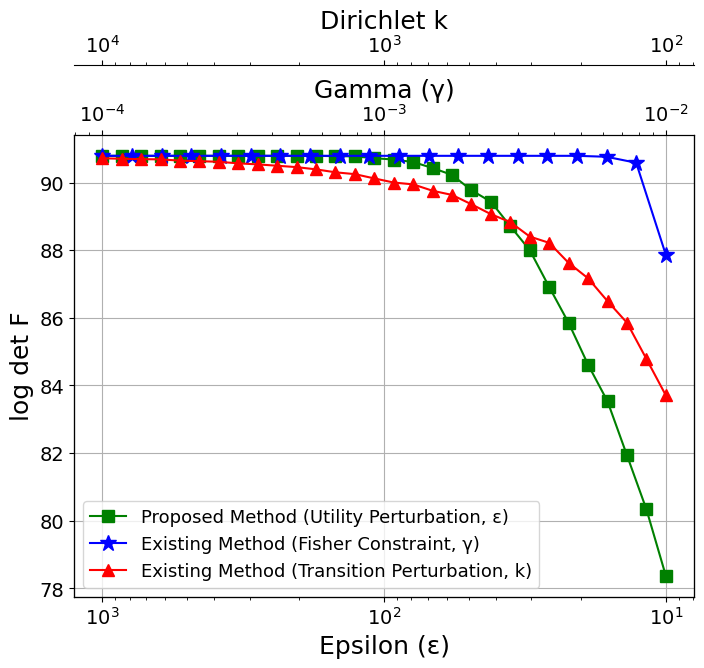}\\
        (a)
    \end{minipage}
    \hfill
    \begin{minipage}{0.4932\linewidth}
        \centering
        \includegraphics[width=\linewidth]{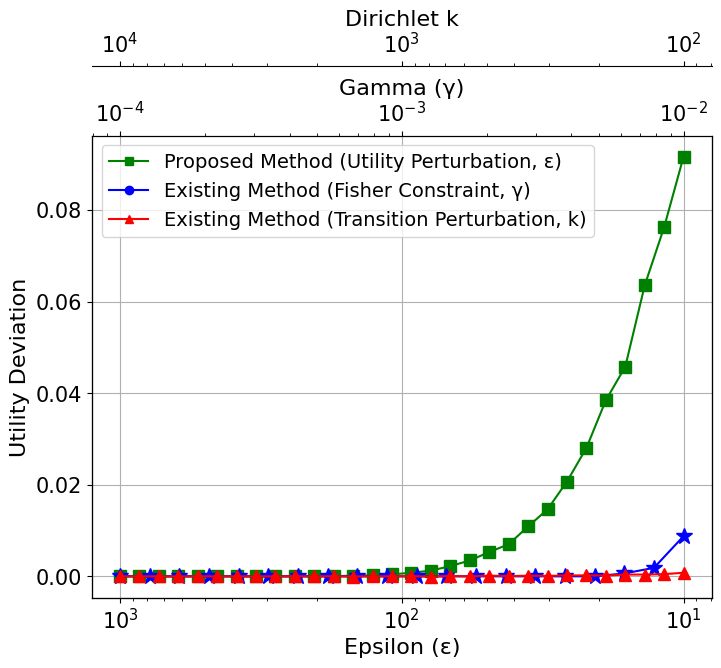}\\
        (b)
    \end{minipage}
    \caption{Comparison of (a) adversarial identifiability
    ($\log\det F$) and (b) utility deviation versus privacy
    parameter $\epsilon$. The proposed method achieves a greater $\log\det F$ reduction
    than~\cite{JainTAES2025,GohariCDC2020} without Fisher
    constraints or transition matrix perturbation, at a
    moderate utility deviation consistent with the bound in
    Theorem~\ref{thm:cost:loss}. The parameter $\epsilon$
    is mapped to $\gamma$ and $k$
    of~\cite{JainTAES2025,GohariCDC2020} to align the privacy budgets for comparison at minimal privacy.}
    \label{fig:privacy_framework}
\end{figure}

The proposed method does incur a moderate utility deviation
relative to the existing approaches, as shown in
Fig.~\ref{fig:privacy_framework}(b); this is the price paid
for achieving a formal $\epsilon$-DistP guarantee without
relying on Fisher information constraints or transition-matrix
perturbation, and is consistent with the analytical bound in
Theorem~\ref{thm:cost:loss}. Operating directly on the
vMF-distributed utility $U$ with the same Laplace-noise
mechanism as WDPCH-SU (Algorithm~\ref{alg:wasserstein-privacy}),
scaled to dimension $n = |\mathcal{S}||\mathcal{A}|$, this
deviation remains bounded and predictable. Across all three
approaches, utility deviation increases with privacy,
confirming the inherent privacy--performance trade-off
characterised in Theorem~\ref{thm:cost:loss}.

\section{Conclusion}\label{sec:conclusion}

In this article, we proposed an online ECCM framework to
conceal the strategic decision-making of a cognitive radar
(CR) against adversarial inference. The radar's utility
function was modeled via a von Mises--Fisher (vMF)
distribution under two complementary paradigms: a static
constrained utility-maximization setting and a dynamic
expected utility-maximization. We adopted a distribution privacy
framework, developing WDPCH-SU and WDPCH-DU for the
static and dynamic settings, respectively, each injecting
Laplace-calibrated noise into the vMF-distributed utility
prior to policy computation, with formal $\epsilon$-DistP
guarantees established for each. Analytical bounds
characterizing the privacy--performance trade-off were
derived, quantifying utility loss (static) and expected
utility deviation (dynamic) as functions of $\epsilon$.
Numerical results show that, in comparison with existing
methodology, WDPCH-SU achieves lower utility loss and
stronger privacy than~\cite{Pattanayak2023}, while
WDPCH-DU attains a greater reduction reduction in adversarial
Fisher information than existing approaches, without requiring explicit Fisher
information constraints~\cite{JainTAES2025} or transition
probability perturbation~\cite{GohariCDC2020}---both with
minimal modification to baseline radar controllers.

\appendix
\renewcommand{\thealgorithm}{\Alph{algorithm}}
\setcounter{algorithm}{0}

In this appendix, we document the algorithms used for adversarial inference 
and privacy-preserving policy synthesis.

Algorithm~\ref{alg:MH} implements the Metropolis--Hastings (MH) MCMC 
procedure used by the adversary to infer the posterior distribution 
$P(\mu \mid \mathcal{D})$ of the radar's cognitive parameter from the 
observed dataset $\mathcal{D} = \{\boldsymbol{\alpha}_k, 
\boldsymbol{\beta}_k^*\}_{k=1}^{N}$, by iteratively sampling candidate 
parameters and accepting or rejecting them via a likelihood ratio test.

\begin{algorithm}[H]
\caption{Metropolis-Hastings to generate posterior distribution}
\label{alg:MH}
\begin{algorithmic}[1]
\Require Observed dataset,
$\mathcal{D} = \{\boldsymbol{\alpha}_k,
\boldsymbol{\beta}_k^*\}_{k=1}^N$ (or, when run on the perturbed
response, $\tilde{\mathcal{D}} = \{\boldsymbol{\alpha}_k,
\tilde{\boldsymbol{\beta}}_k^*\}_{k=1}^N$),
prior distribution $P_0(\theta)$, proposal distribution $Q(\theta'|\theta)$, number of
iterations $T$
\State \textbf{Initialize:} Sample initial parameter $\theta^{(0)} \sim P_0(\theta)$
\For{$t = 1$ to $T$}
    \State Sample candidate $\theta' \sim Q(\theta'|\theta^{(t-1)})$
    \State Compute likelihood ratio
    \label{algo:step:likelihood}
    $r = \min\left(\frac{P(\mathcal{D}|\theta')}
    {P(\mathcal{D}|\theta^{(t-1)})}, 1\right)$
    \State Sample $u \sim \mathrm{Uniform}(0,1)$
    \If{$u < r$}
        \State Accept candidate: $\theta^{(t)} \gets \theta'$
    \Else
        \State Retain previous state: $\theta^{(t)} \gets \theta^{(t-1)}$
    \EndIf
\EndFor
\State \Return $\theta^{(T)}$
\end{algorithmic}
\end{algorithm}
Step~\ref{algo:step:likelihood} of Algorithm~\ref{alg:MH} requires the 
computation of likelihood. The likelihood can be computed numerically as 
follows. Each observation $\mathcal{D}_k = (\boldsymbol{\alpha}_k,
\boldsymbol{\beta}_k^*)$ represents a region 
$\mathcal{U}^{(\mathcal{D}_k)} \subset \mathcal{S}^{n-1}$, such that 
$\boldsymbol{\beta}_k^*$ is the solution of \eqref{eq:LP-unperturbed} 
for all $\mathbf{u}_k \in \mathcal{U}^{(\mathcal{D}_k)}$, i.e.,
$\mathcal{U}^{(\mathcal{D}_k)} := \{\mathbf{u}_k \in \mathcal{S}^{n-1} : 
\boldsymbol{\beta}_k^* \textnormal{ is a solution of 
\eqref{eq:LP-unperturbed}}\}.$
Given the parameter $\theta$, the likelihood of the dataset 
$\mathcal{D}_{1:N}$ is
\begin{equation}\label{Likelihood}
    {\mathbb{P}}({\mathcal{D}_{1:N}}|{\theta}) := \prod_{k=1}^N 
    \mathbb{P}(\mathcal{D}_k|\theta) = \prod_{k=1}^N
    \int_{\mathbf{u}\in\mathcal{U}^{(\mathcal{D}_k)}}
    f(\mathbf{u};{{\theta}})\,d\mathbf{u}.
\end{equation}

Algorithm~\ref{alg:LL} provides a sampling-based procedure to numerically 
evaluate the likelihood $\mathbb{P}(\mathcal{D}_k \mid \theta)$ required 
in the MH acceptance ratio of Algorithm~\ref{alg:MH}, by drawing samples 
$\hat{\mathbf{u}} \sim \mathrm{vMF}(\theta)$, solving the constrained LP 
in~\eqref{eq:LP-unperturbed}, and computing the fraction of samples whose 
optimal response matches the observed $\boldsymbol{\beta}_k^*$ within a 
threshold $\delta_{\mathrm{thr}}$.

\begin{algorithm}[H]
\caption{Algorithm for evaluating likelihood
$\mathbb{P}(\mathcal{D}_k|\theta)$}
\label{alg:LL}
\begin{algorithmic}[1]
\Require $\mathcal{D}_k = (\boldsymbol{\alpha}_k,
\boldsymbol{\beta}_k^*)$ and $\theta$
\State $M \gets$ number of samples
\State $\delta_{thr} \gets$ comparison threshold
\State $n \gets 1$, $\mathrm{count} \gets 0$
\While{$n \leq M$}
    \State $\hat{\mathbf{u}} \sim \mathrm{vMF}(\theta)$,
    $n \gets n + 1$
    \State $\hat{\boldsymbol{\beta}} \gets$
    solve~\eqref{eq:LP-unperturbed} with
    $\mathbf{u} = \hat{\mathbf{u}}$
    \State $d \gets \|\boldsymbol{\beta}_k^* -
    \hat{\boldsymbol{\beta}}\|_2$
    \If{$d \leq \delta_{thr}$}
        \State $\mathrm{count} \gets \mathrm{count} + 1$
    \EndIf
\EndWhile
\State $\mathbb{P}(\mathcal{D}_k|\theta) \gets
\frac{\mathrm{count}}{M}$
\end{algorithmic}
\end{algorithm}

Algorithm~\ref{alg:policywalk} is an MCMC sampling procedure adapted from 
the PolicyWalk algorithm of Ramachandran and Amir~\cite{Ramachandran2007}, 
used by the adversary to perform Bayesian inverse learning of the cognitive 
radar's latent utility $U$ from the observed dataset 
$\mathcal{D} = \{(s_k, a_k)\}_{k=1}^{K}$. At each step, a neighbouring 
utility vector $U'$ is proposed and accepted or rejected via a 
Metropolis--Hastings ratio $P(U',\rho')/P(U,\rho)$, with the policy 
updated through warm-started policy iteration whenever $U'$ violates the 
Bellman optimality condition; the resulting sample mean converges to 
$\hat{U}_{\text{post}} = \mathbb{E}[U \mid \mathcal{D}]$, the adversary's 
estimate of the radar's true utility.

\begin{algorithm}[ht]
\caption{PolicyWalk Sampling for Inverse Learning (Adversary's Algorithm)}
\label{alg:policywalk}
\begin{algorithmic}[1]
\Require $P_{ij}(a)$, rationality $\alpha_X$, dataset
         $\mathcal{D}=\{(s_k,a_k)\}_{k=1}^{K}$, grid bound
         $R_{\max}$, step $\delta$, iterations $N$, burn-in $B$
\State Init $U^{(0)}$ uniformly on grid
       $\mathcal{G}=\{-R_{\max},\dots,R_{\max}\}$;
       $(Q,\pi)\gets\textsc{ValueIteration}(P_{ij}(a),U^{(0)})$
\State $\ell \gets \log P(\mathcal{D}\mid Q,\alpha_X) + \log P(U^{(0)})$
       \Comment{likelihood \& prior as in~\cite{Ramachandran2007}}
\For{$t=1$ \textbf{to} $N$}
    \State Propose $U'$: perturb one random coordinate by $\pm\delta$ (clipped to $[-R_{\max},R_{\max}]$)
    \State Evaluate $Q^\pi_{U'}$ via linear solve under current $\pi$;
           refine via policy iteration only if the greedy policy changes
    \State $\ell' \gets \log P(\mathcal{D}\mid Q',\alpha_X) + \log P(U')$
    \State Accept $(U',Q',\pi',\ell')$ w.p. $\min(1,e^{\ell'-\ell})$; else retain current
    \State Store $U^{(t)} \gets U$
\EndFor
\State \Return $\hat{U}_{\text{post}} = \frac{1}{N-B}\sum_{t=B+1}^{N} U^{(t)}$
\end{algorithmic}
\end{algorithm}

Algorithm~\ref{alg:privacy} perturbs the transition probability matrix of 
the dynamic-utility CR model using Dirichlet noise~\cite{GohariCDC2020}: 
the transition probabilities $P_{ij}(a)$ are privatised by applying a 
Dirichlet mechanism independently to each row, yielding 
$\tilde{P}_{ij}(a) \sim \mathrm{Dirichlet}(k\,P_{ij}(a))$, and the optimal 
radar policy $\tilde{\rho}$ is then synthesised for the privatised MDP 
$\tilde{\mathcal{M}}$ via standard Bellman equations, where the concentration 
parameter $k$ controls the privacy--utility trade-off: smaller $k$ introduces 
higher perturbation and stronger privacy, while larger $k$ preserves the 
original transition structure more closely.

\begin{algorithm}[H]
\caption{Privacy-Preserving Synthesis Algorithm}
\label{alg:privacy}
\begin{algorithmic}[1]
\Require MDP $\mathcal{M} = (\mathcal{S}, \mathcal{A},
P_{ij}(a), \lambda)$, Dirichlet noise parameter $k$
\Ensure Policy $\tilde{\rho}$, value function
$V_{\tilde{\rho}}$
\State Construct privatised transition probabilities:
$\tilde{p}_{ij}(a) \sim \mathrm{Dirichlet}(k\,p_{ij}(a))$
for all $p_{ij}(a) \in \mathcal{P}$
\State Form privatised MDP $\tilde{\mathcal{M}} =
(\mathcal{S}, \mathcal{A}, \bar{\mathcal{P}}, \lambda)$
\State Synthesize policy $\tilde{\rho}$ for
$\tilde{\mathcal{M}}$
\State Compute value function $V_{\tilde{\rho}}$
\State \Return $\tilde{\rho}, V_{\tilde{\rho}}$
\end{algorithmic}
\end{algorithm}

\bibliographystyle{IEEEtran}
\bibliography{references}
\end{document}